\documentclass[11pt,a4paper]{article}

\usepackage[a4paper,margin=25mm]{geometry}
\usepackage{amsmath,amssymb,amsthm}
\usepackage{graphicx}
\usepackage{booktabs}
\usepackage[super]{cite}
\usepackage{xcolor}
\usepackage{placeins}

\usepackage[verbose,hypertexnames=false]{hyperref}

\hypersetup{
  colorlinks=false,
  allbordercolors=blue,
  pdfborderstyle={/S/U/W 1}
}

\makeatletter
\newcommand{\address}[1]{%
  \g@addto@macro\@author{\\[1ex]\normalsize #1}%
}
\makeatother
\newcommand{\catchline}[5]{}
\newcommand{\refcite}[1]{\cite{#1}}
\newcommand{\keywords}[1]{%
  \par\medskip
  \noindent\textbf{Keywords:} #1\par\medskip
}
\newcommand{\tbl}[2]{%
  \centering
  \caption{#1}%
  \vspace{0.5em}%
  #2%
}
\newcommand{\colrule}{\midrule}
\newcommand{\botrule}{\bottomrule}
\date{}

\newtheorem{thm}{Theorem}

\newtheorem{myprop}{Proposition}[section]
\newtheorem{cor}{Corollary}[section]
\newtheorem{defn}{Definition}
\newtheorem{ass}{Assumption}
\newtheorem{rem}{Remark}

\newcommand{\R}{\mathbb R}
\newcommand{\C}{\mathbb C}
\newcommand{\Rew}{\operatorname{Re}}
\newcommand{\Imw}{\operatorname{Im}}
\newcommand{\supp}{\operatorname{supp}}
\newcommand{\rev}[1]{{\color{black}#1}}

\begin{document}

\markboth{M. Kaminaga}{Black--Hole Echo Spectra and Source Dependence}

\catchline{}{}{}{}{}
%%%%%%%%%%%%%%%%%%%%%%%%%%%%%%%%%%%%%%%%%%%%%%%%%%%%%%%%%%%%%%%%%%%%

\title{Black--Hole Echo Resonance Spectra and Source Dependence in a Controlled Transfer--Function Model}

\author{Masahiro Kaminaga}

\address{Tohoku Gakuin University\\
Faculty of Engineering, Department of Information Technology\\
3-1 Shimizukoji, Wakabayashi-ku, Sendai, Miyagi, Japan\\
E-mail: kaminaga@mail.tohoku-gakuin.ac.jp}

\maketitle

% The following history block is normally filled in by the publisher.
%\begin{history}
%\received{Day Month Year}
%\revised{Day Month Year}
%\accepted{Day Month Year}
%\published{Day Month Year}
%\end{history}

\begin{abstract}
\rev{Echo models phenomenologically encode possible near--horizon structure
by replacing the purely ingoing horizon--side condition with an
effective reflecting inner boundary near the would--be horizon.
We study this idea in a controlled transfer--function model consisting
of a compactly supported one--dimensional barrier and a Robin wall at
$x=-L$, where $L>0$ is the cavity length measured in the tortoise
coordinate.
The aim is not to propose a new echo mechanism or to make an
observational claim, but to analyze the standard cavity denominator
in a controlled model with explicit normalizations.
For this model, we prove a local one--zero--per--cell result with an
$O(L^{-2})$ localization error and derive a source--to--observer
factorization which separates the homogeneous poles from the
source--dependent residues.
To test the physical robustness of the mechanism, we also perform
direct numerical calculations for the untruncated axial
Regge--Wheeler potential.
The computed resonances form the same nearly equally spaced comb,
their deviations from the first asymptotic centers are numerically
consistent with $L^{-2}$ scaling, and smooth source profiles modify
the peak weights without changing the homogeneous pole locations.
The rigorous $O(L^{-2})$ theorem remains restricted to the compactly
supported model.}
\end{abstract}


\section{Introduction}
\label{sec:introduction}

The ringdown of a compact object connects gravitational physics, wave propagation, and scattering resonances.
In the standard black--hole picture, late--time ringdown is described by quasinormal modes, namely resonances of a radial scattering problem with a purely ingoing condition at the horizon.
Early references include the works of Vishveshwara, Chandrasekhar and Detweiler, and Leaver; see Refs.~\refcite{Vishveshwara1970,ChandrasekharDetweiler1975,Leaver1985}.
Reviews are given by Kokkotas and Schmidt and by Berti, Cardoso, and Starinets; see Refs.~\refcite{KokkotasSchmidt1999,BertiCardosoStarinets2009}.
A recent account of black--hole spectroscopy is given in Ref.~\refcite{BertiEtAl2025}.

Echo models phenomenologically encode possible near--horizon structure by changing the horizon--side boundary condition.
In the classical black--hole model, the event horizon is not a reflecting surface, and after separation of variables this is represented by a purely ingoing condition on the horizon side.
In models motivated by exotic compact objects, the exterior geometry may remain close to that of a black hole down to a near--horizon region, while the spacetime need not contain a classical event horizon.
At the level of linear waves, such an inner region may be modeled by an effective partially reflecting boundary near the would--be horizon.

A basic motivation for this viewpoint was given by Cardoso, Franzin, and Pani, who emphasized that early ringdown mainly probes the light ring rather than the event horizon~\refcite{CardosoFranzinPani2016}.
The echo picture was then developed for exotic compact objects and possible near--horizon modifications in Ref.~\refcite{CardosoHopperMacedoPalenzuelaPani2016}; see also Ref.~\refcite{CardosoPani2019}.
A tentative observational claim was made in Ref.~\refcite{AbediDykaarAfshordi2017}, whereas a later analysis reported low statistical significance~\refcite{WesterweckEtAl2018}.
Here, echo models are used only as theoretical probes of a frequency--domain scattering problem, and no observational claim is made.

\rev{Echo--like signals are not restricted to horizonless compact objects
or to models with a reflecting surface outside the would--be horizon.
Dong and Stojkovic showed that, in ghost--free massive gravity, a
coupling between gravitational perturbations and scalar hair can
produce a characteristic double--peak effective potential and
gravitational--wave echoes for a black hole that still possesses an
event horizon
\refcite{DongStojkovic2021}.
This example illustrates that an echo--producing cavity should not be
identified uniquely with a horizonless object.
The Robin wall used below is therefore only a controlled realization
of repeated reflection and is not intended as an exhaustive
description of all possible echo mechanisms.}

After separation of variables, the relevant radial equations have the one--dimensional form
$$
\left(- {d^2\over dx^2} + V_\ell(x) - \omega^2\right)u(x)=0,
$$
where $x$ is a tortoise coordinate, the horizon corresponds to $x\to -\infty$, and spatial infinity corresponds to $x\to +\infty$.
The black--hole barrier is of Regge--Wheeler or Zerilli type; see Refs.~\refcite{ReggeWheeler1957,Zerilli1970PRL,Zerilli1970PRD,Chandrasekhar1983}.
For rotating black holes, the Teukolsky and Sasaki--Nakamura formulations are used; see Refs.~\refcite{Teukolsky1972,SasakiNakamura1982,SasakiTagoshi2003}.

\rev{In the rigorous part of this paper, we replace the effective radial
potential by a real compactly supported barrier.
This mathematical cutoff is not a new physical model of the
black--hole potential, but it gives a controlled transfer--function
model in which the scattering normalization, the cavity denominator,
and the long--cavity asymptotics can be treated explicitly.
All frequency windows are compact subsets of $(0,\infty)$, and all
rigorous $O(L^{-2})$ estimates are proved for this compactly supported
model.
Section~\ref{subsec:rw_numerics} gives a direct numerical test for the
untruncated axial Regge--Wheeler potential and determines which parts
of the compactly supported analysis remain quantitatively accurate
for a physically standard black--hole perturbation potential.}

\rev{The cavity interpretation also has earlier antecedents in studies
of ultracompact relativistic stars.
Long--lived axial quasi--normal modes associated with a potential well
inside the Regge--Wheeler barrier were studied by Chandrasekhar and
Ferrari and by Kokkotas
\refcite{ChandrasekharFerrari1991,Kokkotas1994}.
These modes were described as trapped or quasi--stationary states.
Their excitation by prescribed initial perturbations or scattered
particles, and the resulting frequency-- and time--domain responses,
were subsequently investigated in
Refs.~\refcite{Kokkotas1997,TominagaSaijoMaeda1999,FerrariKokkotas2000}.
In particular, periodic spectral peaks produced by waves reflected
between an inner region and an exterior Regge--Wheeler type barrier
were already found in this setting
\refcite{TominagaSaijoMaeda1999}.
The present paper does not claim the existence of this cavity spectrum
or its leading spacing as a new physical observation.}

The Green--function and transfer--function viewpoint is natural for quasinormal--mode excitation \refcite{Leaver1986,NollertPrice1999}.
It is also natural for echo models \refcite{MarkZimmermanDuChen2017,ConklinHoldomRen2018,ConklinHoldom2019}.
In particular, Conklin, Holdom, and Ren emphasized that the echo delay can be searched for through a nearly equally spaced resonance pattern in the frequency domain, rather than only through a detailed time--domain waveform template~\refcite{ConklinHoldomRen2018}.
Analytical echo models for spinning remnants were studied in Ref.~\refcite{MaggioTestaBhagwatPani2019}.
The present work starts from this frequency--domain resonance picture and gives a local analytic realization of it in a half--line scattering model.

\rev{The Fabry--Perot type denominator and the leading large--cavity
spacing are already standard features of trapped--mode and echo
models.
The contribution of the present paper is therefore not to introduce
this physical mechanism.
Our first contribution is to derive the normalized denominator from
the outgoing Jost solution and the Robin wall condition and to
identify precisely the nonvanishing prefactor which relates its zeros
to the resonances of the full half--line problem.
Our second contribution is a local one--zero--per--cell theorem with an
explicit $O(L^{-2})$ localization error in a regular complex frequency
window.
Our third contribution is a Wronskian factorization of the
source--to--observer response which shows at the residue level how a
homogeneous pole can be enhanced, suppressed, or removed for a
particular source.
Finally, we compare these conclusions with direct numerical
calculations for the untruncated Regge--Wheeler potential.}

The source factorization separates the homogeneous resonance spectrum from the spectral response to a particular excitation.
For the inhomogeneous problem we obtain $\psi_\omega=K_L(\omega)D_L(\omega)$, where $K_L$ contains the Wronskian transfer factor and $D_L$ is a source factor.
Thus the same homogeneous pole may give a prominent peak, a suppressed peak, or a canceled pole in a particular response.
Here, source dependence refers only to this residue--level effect in the frequency--domain Green--function response, not to astrophysical source modeling or detectability in gravitational--wave data.

We consider $P_L=-d^2/dx^2+V(x)$ on $[-L,\infty)$, where $V$ is a real compactly supported barrier.
At $x=-L$ we impose $u'(-L)=\gamma u(-L)$, $\gamma\in\mathbb R$, and at $+\infty$ we impose an outgoing condition.
In Schwarzschild coordinates, placing a wall at $x=-L$ corresponds to an effective inner boundary exponentially close to the would--be horizon.
Indeed, from
$$
r_*=r+2M\log\left({r\over 2M}-1\right),
$$
one obtains
$$
r_{\rm wall}
=
2M\left(1+C_* e^{-L/(2M)}(1+o(1))\right),
\qquad L\to\infty,
$$
where $C_*>0$ depends on the additive normalization of $r_*$.

Let $R(\omega)$ be the barrier reflection coefficient and let $\rho_\gamma(\omega)=(i\omega+\gamma)/(i\omega-\gamma)$ be the wall reflection coefficient.
In a regular window where the divided prefactor is nonzero, the homogeneous resonance problem reduces to
$$
1-\rho_\gamma(\omega)R(\omega)e^{2i\omega L}=0.
$$
Here $e^{2i\omega L}$ is the round--trip phase and $Q_\gamma(\omega)=\rho_\gamma(\omega)R(\omega)$ is the total round--trip reflection coefficient.
In a regular frequency window, where $Q_\gamma$ is holomorphic and has no zeros or poles in a complex neighborhood, we fix one branch of $\log Q_\gamma$.
If $x_n=\pi n/L$, the local zero satisfies
$$
\omega_n(L)=x_n+{i\over 2L}\log Q_\gamma(x_n)+O(L^{-2}).
$$
Consequently,
$$
\Rew\omega_{n+1}(L)-\Rew\omega_n(L)
={\pi\over L}+O(L^{-2}),
\qquad
\Imw\omega_n(L)
={1\over 2L}\log |Q_\gamma(x_n)|+O(L^{-2}).
$$
This is the local resonance comb in the long--cavity limit.
For a smoothly cut--off Regge--Wheeler or Zerilli barrier, the same proof applies with the corresponding reflection coefficient.
\rev{For untruncated black--hole potentials, tails change the scattering
phase and the analytic structure.
The $O(L^{-2})$ estimate is proved here only for the compactly
supported model, while the leading round--trip form is tested
numerically for the untruncated axial Regge--Wheeler potential in
Section~\ref{subsec:rw_numerics}.}
%
%Sections~\ref{sec:model}--\ref{sec:asymptotics} derive the resonance equation and prove the local long--cavity asymptotics.
%Sections~\ref{sec:response}--\ref{sec:discussion} discuss source dependence, echo delay, numerical illustrations, and limitations.

\section{The half--line model}
\label{sec:model}

Let $V$ be a real--valued potential on $\mathbb R$.
Throughout this paper we work with a compactly supported potential.
This assumption is stronger than is needed for the physical motivation,
but it makes the mechanism clear.
\begin{ass}
\label{assump:potential}
For some $a>0$, the potential $V$ satisfies
$$
V\in C_0^\infty(\R), \qquad \supp V\subset [0,a].
$$
\end{ass}
The word ``barrier'' refers to the intended physical situation and to the
nonnegative examples used in the numerical section. 
The analytic statements below do not require a sign condition such as $V\geq 0$.
The smoothness assumption is only a convenient sufficient regularity condition.
What is essential here is that $V$ is real and compactly supported: 
outside $[0,a]$ the equation is free, so the Jost solution has plane--wave
representations and the reflection coefficient is meromorphic in $\omega$ and is therefore
holomorphic in a complex neighborhood of each regular window, 
after the possible poles have been excluded.

For $L>0$ and $\gamma\in\R$, define
$$
P_L=-{d^2\over dx^2}+V(x)
$$
on $[-L,\infty)$ with the Robin boundary condition at $x=-L$,
\begin{equation}
\label{eq:robin_boundary}
u'(-L)=\gamma u(-L).
\end{equation}
We call $L$ the cavity length and $\gamma$ the wall parameter.
Figure~\ref{fig:model_schematic} illustrates the geometry of the model.

\begin{figure}[htbp]
\centering
\includegraphics[width=1.00\linewidth]{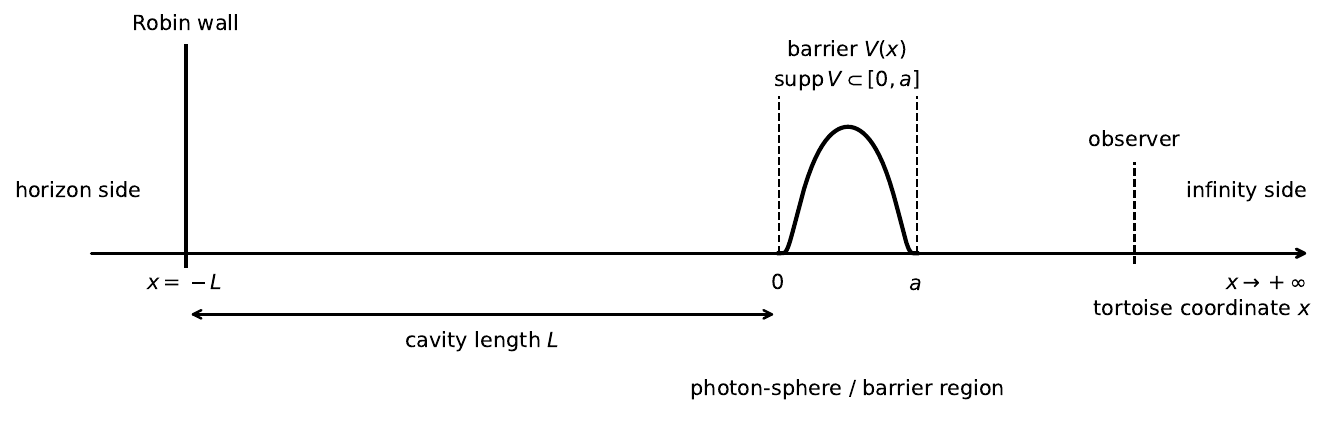}
\caption{
Schematic illustration of the half--line echo model in the tortoise
coordinate. The Robin wall at $x=-L$ and the compactly supported
exterior barrier $V$ form a cavity.
}
\label{fig:model_schematic}
\end{figure}

\begin{defn}
\label{defn:resonance}
A complex number $\omega\neq 0$ is called a resonance of $P_L$ if there exists
a nonzero solution $u$ of
\begin{equation}
\label{eq:eigen_equation}
\left(
-{d^2\over dx^2}+V(x)-\omega^2
\right)u=0
\end{equation}
on $[-L,\infty)$ such that $u$ satisfies \eqref{eq:robin_boundary} and
has the outgoing form
$$
u(x)=Ce^{i\omega x},
        \qquad x>a,
$$
with some constant $C\neq 0$.
\end{defn}
We shall use this definition only away from the threshold $\omega=0$.
With the time dependence $e^{-i\omega t}$, a mode with
$\Imw\omega<0$ decays in time, whereas a mode with $\Imw\omega>0$ grows.
For real $\omega>0$ and real $\gamma$, one has $|\rho_\gamma(\omega)|=1$. 
Hence, if $0<|R(\omega)|<1$, the asymptotic formula below puts the corresponding
 long--cavity resonances in the lower half--plane for large $L$.

\begin{rem}
\label{rem:physical_meaning}
The wall at $x=-L$ models a reflecting inner boundary near the
would--be horizon, while the potential near $x=0$ models the exterior
light--ring barrier, for example a Regge--Wheeler type barrier.
The model is not intended to be a full compact--object model. Rather,
it is a controlled half--line model which isolates the cavity resonance
mechanism and makes the transfer--function analysis explicit.
\end{rem}

\section{Jost solutions and the resonance equation}
\label{sec:jost}

Throughout this section, and in the asymptotic theorem below, we work at nonzero frequency. 
More precisely, the regular frequency windows used in Section~\ref{sec:asymptotics} 
are bounded away from $\omega=0$. 
At $\omega=0$, the two plane waves $e^{i\omega x}$ and
$e^{-i\omega x}$ both reduce to the constant solution $1$.
Therefore, we do not use the formulas for $A(\omega)$ and $B(\omega)$,
which involve division by $\omega$.
Threshold behavior at $\omega=0$ would require a separate analysis.

Let $f_+(x,\omega)$ be the right outgoing Jost solution, defined by
\begin{equation}
\label{eq:jost_equation}
\left(
-{d^2\over dx^2}+V(x)-\omega^2
\right)f_+(x,\omega)=0
\end{equation}
and
\begin{equation}
\label{eq:jost_outgoing}
f_+(x,\omega)=e^{i\omega x},\qquad x>a.
\end{equation}
Since $V=0$ on $(-\infty,0)$, we can write, for $x<0$,
\begin{equation}
\label{eq:jost_left_form}
f_+(x,\omega)=A(\omega)e^{i\omega x}+B(\omega)e^{-i\omega x}.
\end{equation}
For $\omega\neq 0$, the coefficients are determined by the values at
$x=0$:
\begin{eqnarray*}
A(\omega)+B(\omega)
&=&
f_+(0,\omega),
\\
i\omega A(\omega)-i\omega B(\omega)
&=&
f_+'(0,\omega).
\end{eqnarray*}
The initial data defining $f_+$ at $x=a$ and the coefficients of the
differential equation depend holomorphically on $\omega$.
Hence, by the standard analytic dependence theorem for ordinary differential
equations, $f_+(0,\omega)$ and $f_+'(0,\omega)$ are entire functions of $\omega$. 
Solving the above system gives
$$
A(\omega)
=
{1\over 2}
\left\{
f_+(0,\omega)+{f_+'(0,\omega)\over i\omega}
\right\},
\qquad
B(\omega)
=
{1\over 2}
\left\{
f_+(0,\omega)-{f_+'(0,\omega)\over i\omega}
\right\}.
$$
Thus, $A(\omega)$ and $B(\omega)$ are holomorphic away from $\omega=0$.
We define the reflection coefficient by
$$
R(\omega)={B(\omega)\over A(\omega)}
$$
at points where $A(\omega)\neq 0$.
This defines a meromorphic function of $\omega$ away from $\omega=0$.

For real $\omega>0$, the quotient $B(\omega)/A(\omega)$ has the standard scattering meaning.
With the time dependence $e^{-i\omega t}$, the factor $e^{i\omega x}$ in
the left free region is right--moving, while $e^{-i\omega x}$ is left--moving. 
Dividing the outgoing Jost solution by $A(\omega)$ gives
$$
{f_+(x,\omega)\over A(\omega)}=
\begin{cases}
 e^{i\omega x}+{B(\omega)\over A(\omega)}e^{-i\omega x},
        & x<0,\\
 {1\over A(\omega)}e^{i\omega x},
        & x>a.
\end{cases}
$$
Thus, $R(\omega)=B(\omega)/A(\omega)$ is the usual reflection amplitude
for incidence from the left, and $T(\omega)=1/A(\omega)$ is the
corresponding transmission amplitude.

Since $V$ is real, the Wronskian current
$$
J[u](x)=\operatorname{Im}\{\overline{u(x)}u'(x)\}
$$
is conserved for real $\omega$. 
In the scattering normalization this
current is the one--dimensional flux: the wave $e^{i\omega x}$ carries
flux $+\omega$, while $e^{-i\omega x}$ carries flux $-\omega$.
Evaluating this current for $u=f_+$ on the two free regions gives
$$
\omega=\omega\{|A(\omega)|^2-|B(\omega)|^2\}.
$$
Hence, we have
$$
|A(\omega)|^2-|B(\omega)|^2=1,
        \qquad
|T(\omega)|^2+|R(\omega)|^2=1.
$$
In particular, $A(\omega)\neq0$ and $|R(\omega)|<1$ for every real $\omega>0$. 
Also, for real $\gamma$ and real $\omega>0$, one has $|\rho_\gamma(\omega)|=1$.
Therefore, on the real axis in a regular window, $0<|Q_\gamma(\omega)|=|R(\omega)|<1$.

Since any right outgoing solution is a constant multiple of $f_+$,
the resonance condition is obtained by imposing the Robin condition on $f_+$ at $x=-L$.
Substituting \eqref{eq:jost_left_form} into
$$
f_+'(-L,\omega)-\gamma f_+(-L,\omega)=0,
$$
we obtain
\begin{eqnarray}
0
&=&
(i\omega-\gamma)A(\omega)e^{-i\omega L}
-
(i\omega+\gamma)B(\omega)e^{i\omega L}.
\label{eq:boundary_substitution}
\end{eqnarray}
Thus, when $A(\omega)\neq0$ and $i\omega-\gamma\neq0$, the resonance
condition is
\begin{equation}
\label{eq:resonance_equation}
1-\rho_\gamma(\omega)R(\omega)e^{2i\omega L}=0,
\end{equation}
where
\begin{equation}
\label{eq:wall_reflection}
\rho_\gamma(\omega)
=
{i\omega+\gamma\over i\omega-\gamma}.
\end{equation}
For real $\omega>0$ and real $\gamma$, this coefficient has unit modulus.
The special case $\gamma=0$ is the Neumann wall and gives
$\rho_\gamma=1$, while the formal Dirichlet limit
$|\gamma|\to\infty$ gives $\rho_\gamma\to -1$.
The undivided equation \eqref{eq:boundary_substitution} is the basic boundary form of the resonance condition.
In a region where $A(\omega)\neq0$ and $i\omega-\gamma\neq0$, division by $(i\omega-\gamma)A(\omega)e^{-i\omega L}$ gives \eqref{eq:resonance_equation}.
Conversely, multiplying \eqref{eq:resonance_equation} by this nonzero factor gives back \eqref{eq:boundary_substitution}.
Thus the normalized denominator is a convenient representation only in windows where the divided prefactor is holomorphic and nonzero.
If one of the divided factors vanishes, the undivided equation should be used instead.
This distinction will be used in Section~\ref{sec:asymptotics}.
There the local zero theorem is first stated for the normalized denominator, and the identification with full resonances is made under the nonvanishing condition for the prefactor.

\section{Asymptotic location of resonances}
\label{sec:asymptotics}

We study the zeros of the normalized resonance denominator in a fixed real frequency interval.
The argument is local.
Thus we do not need global information on the reflection coefficient.

With the notation $Q_\gamma(\omega)=\rho_\gamma(\omega)R(\omega)$, the normalized resonance equation is
\begin{equation}
\label{eq:resonance_equation_Q}
1-Q_\gamma(\omega)e^{2i\omega L}=0.
\end{equation}
We fix a regular frequency window as follows.
Choose numbers $\alpha_-,\alpha,\beta,\beta_+$ such that $0<\alpha_-<\alpha<\beta<\beta_+$,
and put $I=[\alpha,\beta]$.
The relation between the window $I$, the larger real interval $[\alpha_-,\beta_+]$, and a complex neighborhood $U$ is shown schematically in Fig.~\ref{fig:regular_frequency_window}.
The zeros will be constructed for the real centers
$$
x_n={\pi n\over L}
$$
lying in $I$.
The larger interval $[\alpha_-,\beta_+]$ is not the window in which the centers are chosen.
It only gives a margin for the assumption that the round--trip coefficient has a holomorphic nonzero continuation to a complex neighborhood $U$.

\begin{figure}[t]
\centering
\includegraphics[width=0.82\linewidth]{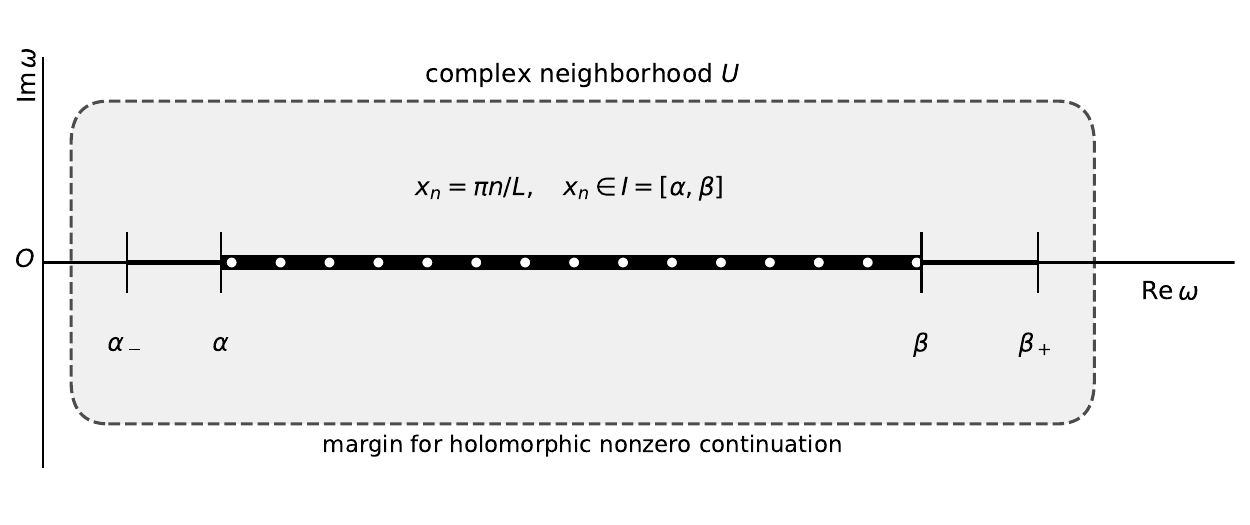}
\caption{
Schematic regular frequency window.
The centers $x_n=\pi n/L$ are chosen in $I=[\alpha,\beta]$, while the larger interval $[\alpha_-,\beta_+]$ gives a margin for a complex neighborhood $U$ where $Q_\gamma$ is holomorphic and nonzero.
}
\label{fig:regular_frequency_window}
\end{figure}

\begin{ass}
\label{assump:Q}
There exists a simply connected complex neighborhood $U$ of the larger interval $[\alpha_-,\beta_+]$ such that $Q_\gamma(\omega)$ is holomorphic and nonzero in $U$.
\end{ass}

Assumption~\ref{assump:Q} is the regularity assumption for the normalized denominator.
It means that the window is chosen away from the threshold $\omega=0$ and away from zeros and poles of the normalized round--trip coefficient $Q_\gamma$.
In particular, it excludes zeros of $R(\omega)$, poles of $R(\omega)$ coming from zeros of $A(\omega)$, and possible zeros and poles of $\rho_\gamma(\omega)$, unless a cancellation occurs in the product.
Thus the theorem below is a local theorem for the normalized equation \eqref{eq:resonance_equation_Q}.
Windows containing such exceptional points should be treated by the undivided boundary equation.

When we identify the zeros of \eqref{eq:resonance_equation_Q} with resonances of the full half--line problem, we also use the following nonvanishing condition:
\begin{equation}
\label{eq:full_prefactor_condition}
A(\omega)\neq 0,
\qquad
i\omega-\gamma\neq 0,
\qquad
\omega\in U.
\end{equation}
This condition is not needed for the proof of Theorem~\ref{thm:main_asymptotics}.
It is needed only when the normalized denominator is compared with the undivided boundary factor.
In the notation of Section~\ref{sec:jost}, that factor is
\begin{equation}
\label{eq:full_boundary_factor_Wtilde}
\widetilde W_L(\omega)
=
(i\omega-\gamma)A(\omega)e^{-i\omega L}
-
(i\omega+\gamma)B(\omega)e^{i\omega L}.
\end{equation}
Under \eqref{eq:full_prefactor_condition}, we have
\begin{equation}
\label{eq:Wtilde_factorization}
\widetilde W_L(\omega)
=
(i\omega-\gamma)A(\omega)e^{-i\omega L}
\left\{
1-Q_\gamma(\omega)e^{2i\omega L}
\right\}
\end{equation}
in $U$.
Since the prefactor in \eqref{eq:Wtilde_factorization} is holomorphic and nonzero in $U$, the zeros of $\widetilde W_L$ and the zeros of the normalized denominator have the same multiplicities in this window.
This is the precise sense in which the normalized equation gives the local resonances of the full half--line problem.

Since $U$ is simply connected and $Q_\gamma$ has no zeros in $U$, there is a holomorphic branch of $\log Q_\gamma(\omega)$ in $U$.
We fix one such branch and write
\begin{equation}
\label{eq:q_def}
q(\omega)=\log Q_\gamma(\omega).
\end{equation}
For real $x\in I$, write
\begin{equation}
\label{eq:a_b_def}
q(x)=a(x)+ib(x),
\end{equation}
where $a(x)=\log |Q_\gamma(x)|$ and $b(x)$ is the argument determined by the chosen branch of $q$.
With this branch, the normalized resonance equation is written locally as
\begin{equation}
\label{eq:log_resonance_equation}
q(\omega)+2i\omega L=2\pi i n,
\qquad n\in\mathbb Z.
\end{equation}
The branch of $q$ is fixed before the integer labels are assigned.
A different branch changes $q$ by $2\pi i m$ and therefore shifts the integer label by $m$.
This is only a relabelling of the same local family, as made explicit in Remark~\ref{rem:branch_choice}.

The next theorem gives the asymptotic location of the local family of zeros associated with $I$.
It is a local theorem.
The uniqueness statement is inside the small disks specified below, not in a global complex strip.

\begin{thm}
\label{thm:main_asymptotics}
Assume Assumption~\ref{assump:Q}.
Then there exist positive constants $L_0$ and $C$, depending on $U$, $I$, and the chosen branch of $q$, but not on $L$ or $n$, such that the following statement holds.

Let $L\geq L_0$ and let $n\in\mathbb Z$ satisfy
$$
x_n={\pi n\over L}\in I.
$$
Then equation \eqref{eq:resonance_equation_Q} has exactly one zero, counted with multiplicity, in the disk
\begin{equation}
\label{eq:disk}
\left|
\omega
-
x_n
-
{i\over 2L}q(x_n)
\right|
\leq
{C\over L^2}.
\end{equation}
This zero is simple.
We denote it by $\omega_n(L)$.
Consequently,
\begin{eqnarray}
\left|
\Rew \omega_n(L)
-
x_n
+
{b(x_n)\over 2L}
\right|
&\leq&
{C\over L^2},
\label{eq:real_part_asymptotics}
\\
\left|
\Imw \omega_n(L)
-
{a(x_n)\over 2L}
\right|
&\leq&
{C\over L^2}.
\label{eq:imag_part_asymptotics}
\end{eqnarray}
If $n$ and $n+1$ satisfy
$$
{\pi n\over L}\in I,
\qquad
{\pi(n+1)\over L}\in I,
$$
then
\begin{equation}
\label{eq:spacing_asymptotics}
\left|
\Rew \omega_{n+1}(L)
-
\Rew \omega_n(L)
-
{\pi\over L}
\right|
\leq
{C\over L^2}.
\end{equation}
The uniqueness asserted here is uniqueness inside the disk \eqref{eq:disk}.
No claim is made that these zeros give all zeros of the denominator in a larger complex strip.
\end{thm}

\begin{proof}
Since $U$ is simply connected and $Q_\gamma$ is holomorphic and nonzero in $U$, the chosen branch $q=\log Q_\gamma$ is holomorphic in $U$.
Choose $r_1>0$ such that
$$
\{\omega\in\C:\operatorname{dist}(\omega,I)\leq r_1\}
\subset U.
$$
Here, $\operatorname{dist}(\omega,I)=\inf_{x\in I}|\omega-x|$ is the Euclidean distance from $\omega$ to the interval $I$.
Then there is a constant $M\geq1$ such that
\begin{equation}
\label{eq:q_bound}
|q(\omega)|+|q'(\omega)|\leq M
\end{equation}
on this neighborhood.
Choose $C\geq1$ such that
$$
C\geq M^2+{M\pi\over2}.
$$
The constant $L_0$ will be chosen sufficiently large in what follows.
Fix $L\geq L_0$ and an integer $n$ such that
$$
x_n={\pi n\over L}\in I.
$$
We look for a zero in an $O(L^{-1})$ neighborhood of $x_n$, and therefore write
$$
\omega=x_n+{z\over L}.
$$
Since $2iLx_n=2\pi i n$, the logarithmic equation
$$
q(\omega)+2iL\omega=2\pi i n
$$
is equivalent to
$$
q\left(x_n+{z\over L}\right)+2iz=0.
$$
Thus, it is equivalent to $z=F(z)$, where
\begin{equation}
\label{eq:F_def}
F(z)
=
{i\over2}
q\left(x_n+{z\over L}\right).
\end{equation}
Let
$$
B=\{z\in\C: |z|\leq M\}.
$$
For $L_0$ sufficiently large, $x_n+z/L$ lies in the above $r_1$--neighborhood for all $n$ with $x_n\in I$ and all $z\in B$.
Hence
$$
|F(z)|\leq {M\over2}\leq M,
\qquad z\in B,
$$
and therefore $F(B)\subset B$.
Moreover, for $z,w\in B$, the line segment joining $z$ and $w$ is contained in $B$.
Since
$$
F'(z)
=
{i\over 2L}
q'\left(x_n+{z\over L}\right),
$$
\eqref{eq:q_bound} gives
$$
|F'(z)|\leq {M\over 2L},
\qquad z\in B.
$$
Thus, we have
$$
|F(z)-F(w)|
\leq
{M\over2L}|z-w|,
\qquad z,w\in B.
$$
We increase $L_0$, if necessary, so that $M<2L_0$.
Then, for every $L\geq L_0$, the map $F$ is a contraction on $B$.
Hence there is a unique $z_n(L)\in B$ satisfying $z_n(L)=F(z_n(L))$.

Define
\begin{equation}
\label{eq:omega_from_z}
\omega_n(L)=x_n+{z_n(L)\over L}.
\end{equation}
The fixed point equation gives
$$
q(\omega_n(L))+2iL(\omega_n(L)-x_n)=0.
$$
Since $2iLx_n=2\pi i n$, we have
$$
q(\omega_n(L))+2iL\omega_n(L)=2\pi i n.
$$
Therefore
$$
Q_\gamma(\omega_n(L))e^{2iL\omega_n(L)}=1.
$$
Thus $\omega_n(L)$ is a zero of the denominator in \eqref{eq:resonance_equation_Q}.

We next compare this zero with the first approximation.
Since
$$
F(0)={i\over2}q(x_n),
$$
the contraction estimate gives
$$
\left|
z_n(L)-{i\over2}q(x_n)
\right|
=
|F(z_n(L))-F(0)|
\leq
{M\over2L}|z_n(L)|
\leq
{M^2\over2L}.
$$
Thus
\begin{equation}
\label{eq:omega_expansion_sharp}
\left|
\omega_n(L)
-
x_n
-
{i\over 2L}q(x_n)
\right|
\leq
{M^2\over2L^2}
\leq
{C\over L^2}.
\end{equation}
This proves the asserted first approximation and also shows that $\omega_n(L)$ belongs to the disk \eqref{eq:disk}.

We now prove uniqueness in the disk.
Let $\omega$ be any zero of \eqref{eq:resonance_equation_Q} satisfying \eqref{eq:disk}.
For $L_0$ sufficiently large, this disk is contained in the $r_1$--neighborhood of $I$, and hence in $U$.
Therefore the fixed branch $q=\log Q_\gamma$ is defined at $\omega$, and $e^{q(\omega)}=Q_\gamma(\omega)$.
Since $\omega$ is a zero of \eqref{eq:resonance_equation_Q}, we have $Q_\gamma(\omega)e^{2iL\omega}=1$.
Thus $q(\omega)+2iL\omega=2\pi i k$ for some $k\in\mathbb Z$.
Put $\omega_n^{(0)}=x_n+iq(x_n)/(2L)$.
Since $2iL\omega_n^{(0)}=2\pi i n-q(x_n)$, we obtain
\begin{eqnarray}
\label{eq:uniqueness_integer_identity}
q(\omega)+2iL\omega-2\pi i n
&=&
\{q(\omega)-q(x_n)\}
+
2iL(\omega-\omega_n^{(0)}).
\end{eqnarray}
For $L_0$ sufficiently large, the line segment joining $x_n$ and $\omega$ is contained in the $r_1$--neighborhood of $I$.
Therefore
\begin{eqnarray*}
|q(\omega)-q(x_n)|
&\leq&
M|\omega-x_n|
\\
&\leq&
M\left\{
{1\over 2L}|q(x_n)|+{C\over L^2}
\right\}
\\
&\leq&
{M^2\over2L}+{MC\over L^2}.
\end{eqnarray*}
Also, since $|\omega-\omega_n^{(0)}|\leq C/L^2$,
$$
2L|\omega-\omega_n^{(0)}|\leq {2C\over L}.
$$
It follows from \eqref{eq:uniqueness_integer_identity} that
$$
\left|q(\omega)+2iL\omega-2\pi i n\right|
\leq
{M^2\over2L}+{MC\over L^2}+{2C\over L}.
$$
The right hand side is smaller than $\pi$ for $L_0$ sufficiently large.
Since $q(\omega)+2iL\omega=2\pi i k$, we have $2\pi |k-n|<\pi$.
Therefore $k=n$.

Set $z=L(\omega-x_n)$.
Then the equation with $k=n$ gives
$$
q\left(x_n+{z\over L}\right)+2iz=0,
$$
or equivalently $z=F(z)$.
Since $x_n+z/L=\omega$
belongs to the $r_1$--neighborhood of $I$, the bound \eqref{eq:q_bound} gives
$$
|z|
=
{1\over2}
\left|q\left(x_n+{z\over L}\right)\right|
\leq
{M\over2}
\leq M.
$$
Hence $z\in B$, and the uniqueness of the fixed point of $F$ gives $z=z_n(L)$.
Therefore $\omega=\omega_n(L)$.
The zero is simple.
Indeed, if
$$
H_L(\omega)=1-Q_\gamma(\omega)e^{2iL\omega},
$$
then at a zero of $H_L$,
\begin{eqnarray*}
H_L'(\omega)
&=&
-
\left\{
Q_\gamma'(\omega)+2iLQ_\gamma(\omega)
\right\}
e^{2iL\omega}
\\
&=&
-\left\{
{Q_\gamma'(\omega)\over Q_\gamma(\omega)}
+
2iL
\right\}
\\
&=&
-\{q'(\omega)+2iL\}.
\end{eqnarray*}
At $\omega=\omega_n(L)$, we have $|q'(\omega_n(L))|\leq M$.
Hence $H_L'(\omega_n(L))\neq0$ for $L_0$ sufficiently large.
By \eqref{eq:a_b_def}, we have $q(x)=a(x)+ib(x)$ for $x\in I$.
Taking real and imaginary parts in \eqref{eq:omega_expansion_sharp}, we obtain \eqref{eq:real_part_asymptotics} and \eqref{eq:imag_part_asymptotics}.
Finally, since $|b'(x)|\leq M$ on $I$, \eqref{eq:omega_expansion_sharp} gives
\begin{eqnarray*}
\left|
\Rew \omega_{n+1}(L)
-
\Rew \omega_n(L)
-
{\pi\over L}
\right|
&\leq&
{M^2\over L^2}
+
{M\over2L}|x_{n+1}-x_n|
\\
&=&
{M^2\over L^2}
+
{M\pi\over2L^2}\leq{C\over L^2}.
\end{eqnarray*}
This proves \eqref{eq:spacing_asymptotics}.
\end{proof}

\begin{cor}
\label{cor:full_resonance_window}
Assume Assumption~\ref{assump:Q} and the additional condition \eqref{eq:full_prefactor_condition}.
Then the zeros constructed in Theorem~\ref{thm:main_asymptotics} are exactly the local zeros of the full boundary factor $\widetilde W_L$ in the disks \eqref{eq:disk}.
In particular, they are local resonances of the half--line problem in this window.
They are simple zeros of $\widetilde W_L$.
\end{cor}

\begin{proof}
By \eqref{eq:Wtilde_factorization}, the full boundary factor is the product of the normalized denominator and the prefactor
$$
(i\omega-\gamma)A(\omega)e^{-i\omega L}.
$$
Under \eqref{eq:full_prefactor_condition}, this prefactor is holomorphic and nonzero in $U$.
Therefore the zeros of $\widetilde W_L$ and the zeros of the normalized denominator have the same multiplicities in the window.
Theorem~\ref{thm:main_asymptotics} gives exactly one simple zero of the normalized denominator in each disk \eqref{eq:disk}.
Hence the same statement holds for $\widetilde W_L$.
By the derivation of the resonance condition in Section~\ref{sec:jost}, the zeros of $\widetilde W_L$ are the resonances of the half--line problem.
\end{proof}

\begin{rem}
Theorem~\ref{thm:main_asymptotics} is a theorem for the normalized equation \eqref{eq:resonance_equation_Q}.
Corollary~\ref{cor:full_resonance_window} is the corresponding statement for the full half--line resonance problem.
This separation is useful because the normalized coefficient $Q_\gamma$ may be regular even when a divided factor in the undivided boundary equation is not suitable for normalization.
Near such a point one should return to the undivided factor $\widetilde W_L$.
\end{rem}

\begin{rem}
The theorem is local in frequency.
For each admissible $x_n=\pi n/L\in I$, it constructs one simple zero in an $O(L^{-2})$ neighborhood of the first approximation
$$
\omega_n^{(0)}=x_n+{i\over2L}q(x_n).
$$
It does not claim that these zeros give all zeros of $1-Q_\gamma(\omega)e^{2i\omega L}$ in any larger complex neighborhood of $I$.
It also does not give a global description of all resonances of the half--line problem.
Zeros and poles of $Q_\gamma$ are excluded.
Near such points the branch $q=\log Q_\gamma$ is not available, and a separate local analysis is needed.
The natural object for such an analysis is the undivided boundary factor $\widetilde W_L$.
\end{rem}

\begin{rem}
\label{rem:branch_choice}
The choice of the branch of $q=\log Q_\gamma$ only affects the integer labelling of the local family.
If the branch is changed to $q+2\pi i m$, with $m\in\mathbb Z$, then $\omega_n^{(0)}$ is shifted by $-\pi m/L$ in its real part.
This corresponds to replacing the label $n$ by $n-m$, up to the $O(L^{-2})$ error already included in the asymptotic formula.
\end{rem}

\begin{rem}
\label{rem:decay_rate}
For real $\omega>0$ and real $\gamma$, one has $|\rho_\gamma(\omega)|=1$.
Thus, on the real axis, $|Q_\gamma(\omega)|=|R(\omega)|$.
By the imaginary--part estimate in Theorem~\ref{thm:main_asymptotics},
$$
\Imw \omega_n(L)
=
{1\over 2L}\log |Q_\gamma(x_n)|+O(L^{-2}).
$$
Hence, if $|R(x)|\leq 1-\delta$ on $I$ for some $\delta>0$, then the zeros constructed in Theorem~\ref{thm:main_asymptotics} lie in the lower half--plane for all sufficiently large $L$.
Their imaginary parts are bounded above by $-c/L$ for some $c>0$.
Under the additional condition \eqref{eq:full_prefactor_condition}, the same statement applies to the corresponding resonances of the full half--line problem.
\end{rem}

\begin{cor}
\label{cor:general_wall}
Let $\rho_{\rm wall}(\omega)$ be a prescribed wall reflection coefficient.
Suppose that, for the same $\alpha_-,\alpha,\beta,\beta_+$, there exists a simply connected complex neighborhood $U_{\rm wall}$ of $[\alpha_-,\beta_+]$ such that
$$
Q_{\rm wall}(\omega)=\rho_{\rm wall}(\omega)R(\omega)
$$
is holomorphic and nonzero in $U_{\rm wall}$.
Then the conclusions of Theorem~\ref{thm:main_asymptotics} hold for the equation
$$
1-Q_{\rm wall}(\omega)e^{2i\omega L}=0
$$
with $Q_\gamma$ and $q=\log Q_\gamma$ replaced by $Q_{\rm wall}$ and $q_{\rm wall}=\log Q_{\rm wall}$.
\end{cor}

\begin{proof}
The proof of Theorem~\ref{thm:main_asymptotics} uses only the existence of a holomorphic branch of the logarithm of the round--trip coefficient in a regular complex neighborhood of $I$.
Since $Q_{\rm wall}$ is holomorphic and nonzero in the simply connected neighborhood $U_{\rm wall}$, such a branch exists for $Q_{\rm wall}$.
The same contraction argument therefore applies with $Q_{\rm wall}$ in place of $Q_\gamma$.
\end{proof}

\begin{rem}
Corollary~\ref{cor:general_wall} is a local frequency--domain transfer--function statement.
A prescribed $\rho_{\rm wall}$ may or may not arise from a local boundary condition at a wall.
If $|\rho_{\rm wall}(x)|<1$ on the real axis, then the leading imaginary part contains the additional damping term
$$
{1\over 2L}\log |\rho_{\rm wall}(x)|.
$$
Thus perfect Robin reflectivity is not essential for the local spacing statement.
Physical restrictions on a prescribed wall reflectivity, such as passivity or causality, are separate questions.
They are not needed for the local resonance--comb theorem proved here.
\end{rem}

%%%%

\section{Source--to--observer response and source dependence}
\label{sec:response}

The preceding sections concern the homogeneous resonances.
For an inhomogeneous excitation, the observed frequency--domain response
also contains a source factor.
We recall the Wronskian form of the source--to--observer response:
$$
\psi_\omega=K_L(\omega)D_L(\omega),
$$
where $K_L$ is the full transfer factor and $D_L$ is determined by the source.

\rev{Source--dependent excitation of relativistic stellar modes has
been studied using both frequency--domain and time--domain perturbation
equations
\refcite{Kokkotas1997,TominagaSaijoMaeda1999,FerrariKokkotas2000}.
Those works use physically specified perturbations or particle
trajectories and calculate the resulting waveforms or energy spectra.
The present factor $D_L(\omega)$ should be understood as an abstract
Wronskian realization of the same general distinction between the
homogeneous mode spectrum and its excitation.
We do not attempt to reproduce a particular particle trajectory or
astrophysical initial perturbation.
Instead, we isolate the analytic dependence of each pole residue on a
compactly supported source.}

\rev{A recent concrete example of source--dependent excitation is the
plunge spectrum toward a black--hole mimicker studied by Nair
\refcite{Nair2026}.
In that setting, the low--frequency energy spectrum contains a comb of
sharp resonances associated with the real parts of the mimicker
quasinormal frequencies, while a further qualitative departure from
the black--hole spectrum appears above a threshold frequency.
The plunge source thus provides a physical setting in which cavity
resonances are selectively excited and become visible in the observed
spectrum.
In the notation of the present paper, this is analogous to a
physically specified source factor weighting the poles of the
homogeneous transfer function.
We do not claim a quantitative identification between the two models,
but this example supports the general distinction between homogeneous
pole locations and source--dependent spectral weights.}

Let $\varphi_L(x,\omega)$ be the solution of the homogeneous equation
\begin{equation}
\label{eq:left_solution}
\left(
-{d^2\over dx^2}+V(x)-\omega^2
\right)\varphi_L(x,\omega)=0
\end{equation}
satisfying $\varphi_L(-L,\omega)=1$ and
$\varphi_L'(-L,\omega)=\gamma$.
Let $f_+(x,\omega)$ be the outgoing Jost solution on the right.
We define
\begin{equation}
\label{eq:W_def}
W_L(\omega)
=
\varphi_L(x,\omega)f_+'(x,\omega)
-
\varphi_L'(x,\omega)f_+(x,\omega).
\end{equation}
This Wronskian is independent of $x$.
By evaluating it at $x=-L$, we get
\begin{equation}
\label{eq:W_boundary}
W_L(\omega)=f_+'(-L,\omega)-\gamma f_+(-L,\omega).
\end{equation}
Thus, the zeros of $W_L$ are precisely the resonances.
The boundary factor $\widetilde W_L$ used in Section~\ref{sec:asymptotics} is the same function as $W_L$.
In a region where $A(\omega)\neq0$ and $i\omega-\gamma\neq0$,
\eqref{eq:boundary_substitution} gives
\begin{equation}
\label{eq:W_factor}
W_L(\omega)
=
(i\omega-\gamma)A(\omega)e^{-i\omega L}
\left\{
1-\rho_\gamma(\omega)R(\omega)e^{2i\omega L}
\right\}.
\end{equation}
%%%
It is useful to distinguish the full transfer factor from the normalized cavity part.
We define
\begin{equation}
\label{eq:K_cav_def}
K_L^{\rm cav}(\omega)
={1\over
1-\rho_\gamma(\omega)R(\omega)e^{2i\omega L}}.
\end{equation}
This factor contains only the cavity denominator.
On the other hand, the Green function formula contains the full Wronskian.
In a regular frequency window where the prefactor
$(i\omega-\gamma)A(\omega)e^{-i\omega L}$ in \eqref{eq:W_factor} is
holomorphic and nonzero, the two factors have the same poles in that window.
We define the full transfer factor by
\begin{equation}
\label{eq:K_full_def}
K_L(\omega)=-{1\over W_L(\omega)}.
\end{equation}
Using \eqref{eq:W_factor}, we obtain
\begin{equation}
\label{eq:K_full_cav_relation}
K_L(\omega)
=
-{e^{i\omega L}\over (i\omega-\gamma)A(\omega)}
K_L^{\rm cav}(\omega).
\end{equation}
Hence, in the regular window, $K_L$ and $K_L^{\rm cav}$ have the same
poles and differ only by a holomorphic nonvanishing prefactor.
The former is the full transfer factor in the source--to--observer formula, 
while the latter is the normalized cavity factor used to visualize the resonance comb.

Now consider the inhomogeneous problem
\begin{equation}
\label{eq:inhomogeneous}
\left(
-{d^2\over dx^2}+V(x)-\omega^2
\right)u(x,\omega)
=
S(x,\omega)
\end{equation}
on $[-L,\infty)$, with the Robin condition at $x=-L$ and the outgoing condition at $+\infty$.
We assume for simplicity that $S(x,\omega)$ is compactly supported in $x$ 
and holomorphic in $\omega$ in the frequency region considered here.
%%%% 
The Green function is
\begin{equation}
\label{eq:green_function}
G_L(x,y;\omega)
=
-{1\over W_L(\omega)}
\begin{cases}
\varphi_L(x,\omega)f_+(y,\omega), & x<y,\\
\varphi_L(y,\omega)f_+(x,\omega), & y<x.
\end{cases}
\end{equation}
Since $x$ is taken to the right of $\operatorname{supp} S$, the second line
of \eqref{eq:green_function} applies throughout the integral.
For $x$ to the right of both the supports of $V$ and $S$, we have
$y<x$ on the support of $S$ and $f_+(x,\omega)=e^{i\omega x}$. 
Hence, the Green representation gives
$$
u(x,\omega)
=
K_L(\omega)
\left\{
\int_{-L}^{\infty}
\varphi_L(y,\omega)S(y,\omega)\,dy
\right\}
e^{i\omega x}.
$$
Since $u(x,\omega)=\psi_\omega e^{i\omega x}$, we obtain
\begin{equation}
\label{eq:response_factorization}
\psi_\omega=K_L(\omega)D_L(\omega),
\end{equation}
where
\begin{equation}
\label{eq:D_def}
D_L(\omega)
=
\int_{-L}^\infty
\varphi_L(y,\omega)S(y,\omega)\,dy.
\end{equation}

The source factor depends on where and how the system is excited.
If the source is supported in the free part of the cavity, $-L<y<0$, then
\begin{equation}
\label{eq:phi_free_cavity}
\varphi_L(y,\omega)
=\cos\{\omega(y+L)\}+
{\gamma\over\omega}\sin\{\omega(y+L)\}.
\end{equation}
Thus, separated source components in the cavity can interfere directly in $D_L$.
If the source is near the barrier or outside the barrier, the same formula \eqref{eq:D_def} applies, 
but $\varphi_L$ also contains the barrier scattering information.
The pole locations are unchanged by these choices; the residues are not.
%%%% 
\begin{myprop}
\label{prop:visibility}
Let $\omega_j$ be a simple zero of $W_L$. Assume that $D_L$ is
holomorphic near $\omega_j$. Then the response $\psi_\omega$ has a simple
pole at $\omega_j$ if and only if
$$
D_L(\omega_j)\neq0.
$$
In this case the residue is
\begin{equation}
\label{eq:residue_formula}
\operatorname{Res}_{\omega=\omega_j}\psi_\omega
=
-{D_L(\omega_j)\over W_L'(\omega_j)}.
\end{equation}
If $D_L(\omega_j)=0$, then the pole is removable in the
source--to--observer response.
\end{myprop}

\begin{proof}
Since $\omega_j$ is a simple zero of $W_L$, we can write
$$
W_L(\omega)=(\omega-\omega_j)G(\omega),
        \qquad G(\omega_j)=W_L'(\omega_j)\neq0,
$$
near $\omega_j$, where $G$ is holomorphic.
By \eqref{eq:response_factorization} and \eqref{eq:K_full_def},
$$
\psi_\omega
=-{D_L(\omega)\over (\omega-\omega_j)G(\omega)}.
$$
This has a simple pole exactly when $D_L(\omega_j)\neq0$, 
and then its residue is
$$
-{D_L(\omega_j)\over G(\omega_j)}
=
-{D_L(\omega_j)\over W_L'(\omega_j)}.
$$
If $D_L(\omega_j)=0$, then, since $D_L$ is holomorphic, the factor
$\omega-\omega_j$ divides $D_L(\omega)$ at least once near $\omega_j$.
Hence, the singularity of $\psi_\omega$ at $\omega_j$ is removable.
This proves the assertion.
\end{proof}
%%%% 
In a regular frequency window, we write
$$
P(\omega)=(i\omega-\gamma)A(\omega)e^{-i\omega L},
        \qquad
H_L(\omega)=1-Q_\gamma(\omega)e^{2i\omega L}.
$$
Then, $W_L=P H_L$, and $P$ is holomorphic and nonzero in the window under consideration.
At a simple zero $\omega_j$ of $H_L$, we have
\begin{equation}
\label{eq:residue_regular_window}
\operatorname{Res}_{\omega=\omega_j}\psi_\omega
={D_L(\omega_j)
\over
P(\omega_j)\{q'(\omega_j)+2iL\}}.
\end{equation}
Thus, away from exceptional normalization points, the residue is governed
by the source factor, with the main denominator factor being of order
$L$, up to nonzero Jost and wall prefactors.

\begin{rem}
Proposition~\ref{prop:visibility} is the elementary analytic mechanism behind spectral visibility.
The pole location is fixed by the homogeneous problem, 
but the residue contains the source factor.
Indeed, near a simple zero $\omega_j$ of $W_L$, the Taylor expansions give
$$
W_L(\omega)
=
W'_L(\omega_j)(\omega-\omega_j)
+
O((\omega-\omega_j)^2),
\qquad
D_L(\omega)
=
D_L(\omega_j)
+
O(\omega-\omega_j).
$$
Hence
$$
\psi_\omega
=
-
{D_L(\omega_j)\over W'_L(\omega_j)}
{1\over \omega-\omega_j}
+
\hbox{regular terms}.
$$
If $\omega_j=\Omega_j-i\Gamma_j$ with $\Gamma_j>0$, then for real $\omega$,
$$
|\omega-\omega_j|
=
\sqrt{(\omega-\Omega_j)^2+\Gamma_j^2}.
$$
Thus, if the regular background and the contributions of other nearby poles do not dominate
near $\Omega_j$, the resonant contribution gives a Lorentzian--type peak, with amplitude
peak size of order
$$
{|D_L(\omega_j)|\over |W'_L(\omega_j)|\,\Gamma_j}.
$$
Thus, a resonance may be weakly visible because the source factor is small, 
because the pole is far from the real axis, or because the background masks the peak.
If $D_L(\omega_j)=0$, the pole is canceled in the response.
\end{rem}

\begin{rem}
\label{rem:visibility_not_location}
The zeros of $W_L$ depend only on the operator and boundary conditions;
the source changes the residues.
Thus, $K_L^{\rm cav}$ displays the cavity comb, 
while $K_LD_L$ determines which part of that comb is visible for a particular excitation.
\end{rem}

%%%% 
\begin{rem}
For an idealized finite superposition of point sources in the free part of the cavity,
\begin{equation}
S(y,\omega)=\sum_m c_m(\omega)\delta(y-y_m),
        \qquad -L<y_m<0,
\end{equation}
we obtain, away from $\omega=0$,
\begin{equation}
D_L(\omega)=
\sum_m c_m(\omega)
\left[
\cos\{\omega(y_m+L)\}
+{\gamma\over\omega}\sin\{\omega(y_m+L)\}
\right].
\end{equation}
The delta functions are distributional idealizations of narrow smooth compactly supported pulses. 
The formula shows how destructive interference can make $D_L(\omega_j)$ small or zero.
Fixing $y_m$ while varying $L$ is different from fixing the distance
$d_m=y_m+L$ from the wall. In the latter case $y_m=-L+d_m$ depends on
$L$, and the phase factors in $D_L(\omega)$ have a different $L$--dependence.
\end{rem}

%%%% 
\section{Relation with echo time delay}
\label{sec:echo_delay}

In units with $c=1$, the cavity length is approximately $L$, so a wave packet
trapped between the wall and the potential barrier returns after a time
$\Delta t_{\rm echo}\sim 2L$.
For the local resonance family constructed in Theorem~\ref{thm:main_asymptotics},
one has $\Delta \Rew\omega\sim \pi/L$, and therefore
$\Delta t_{\rm echo}\,\Delta \Rew\omega\sim2\pi$.
Equivalently, in ordinary frequency, $\Delta f\sim 1/(2L)$ and
$\Delta t_{\rm echo}\sim1/\Delta f$.
This frequency--domain relation between the echo delay and the resonance
spacing is the one used in echo--spectrum searches and transfer--function
descriptions; see, for example,
Refs.~\refcite{ConklinHoldomRen2018,ConklinHoldom2019}.
The point of the present derivation is that the same relation follows
from the zeros of the transfer denominator.
It is a property of the long cavity, not a special feature of one chosen waveform.
More precisely, for real $x\in I$, we write
$$
q(x)=\log Q_\gamma(x)=a(x)+ib(x).
$$
The real part of the local resonance is determined, to leading order, by
the phase condition
$$
2Lx+b(x)\simeq2\pi n.
$$
The left hand side may be viewed as the total round--trip phase. Put
$$
\Theta(x)=2Lx+b(x).
$$
If $x$ and $x+\Delta x$ are neighboring approximate real parts, then
$$
\Theta(x+\Delta x)-\Theta(x)\simeq2\pi.
$$
Taylor's formula gives
$$
2\pi
\simeq
\{2L+b'(x)\}\Delta x
+
{1\over2}b''(x)(\Delta x)^2+\cdots.
$$
Since the leading spacing is of order $L^{-1}$, the quadratic term is
of lower order. Thus, to the first phase--corrected order,
$$
\Delta x
\simeq
{2\pi\over 2L+b'(x)}
=
{\pi\over L}
-
{\pi b'(x)\over 2L^2}
+O(L^{-3}).
$$
Hence the effective round--trip group delay is formally
$$
T_{\rm rt}(x)\simeq \Theta'(x)=2L+b'(x),
$$
and the local spacing is corrected by the frequency dependence of the
round--trip phase:
$$
\Delta \Rew\omega
\simeq
{2\pi\over 2L+b'(x)}.
$$
For the actual resonances, this phase correction should be understood
together with the $O(L^{-2})$ localization error in
Theorem~\ref{thm:main_asymptotics}. The leading relation is obtained when
the cavity length is large and the phase variation of the reflection
coefficient is a lower order correction.

%%%%
\section{Numerical illustration}
\label{sec:numerics}

\rev{\subsection{Compactly supported model}
\label{subsec:compact_numerics}}

\rev{We first give numerical illustrations for the compactly supported
barrier used in the rigorous analysis.
The purpose of this subsection is to verify the local resonance
construction and the source--factor mechanism under the assumptions
of Theorem~\ref{thm:main_asymptotics}.}
We take
$$
V(x)=
\begin{cases}
V_0\exp\left\{1-{a^2\over 4x(a-x)}\right\}, & 0<x<a,\\
0, & x\leq0 \hbox{ or } x\geq a.
\end{cases}
$$
Then $V\in C^\infty_0(\R)$, the support is contained in $[0,a]$, and
$V(a/2)=V_0$. Unless otherwise stated, we use
$$
L=60.0,\qquad \gamma=0.0,\qquad V_0=2.0,\qquad a=1.0,
$$
and display the real frequency interval
$$
0.30\leq \Rew\omega\leq2.00.
$$

For each nonzero complex frequency $\omega$, the outgoing Jost solution
was computed by integrating backwards from $x=a$ to $x=0$ with the outgoing
data in \eqref{eq:jost_outgoing}.
The coefficients $A(\omega)$ and $B(\omega)$ were then recovered from
the formulas in Section~\ref{sec:jost}.
We then set $R(\omega)=B(\omega)/A(\omega)$ and computed the zeros of
$$
F_L(\omega)
=
1-\rho_\gamma(\omega)R(\omega)e^{2i\omega L}.
$$
The numerical integration was done by an adaptive Runge--Kutta method.
The tolerance settings were kept fixed for all displayed runs, and the reported digits did not change when the integration tolerances were tightened.
The zeros were found by a Newton--type iteration for the real and imaginary parts of $F_L$, starting from the asymptotic centers
$$
\omega_n^{(0)}
=
{\pi n\over L}
+
{i\over 2L}
\log Q_\gamma\left({\pi n\over L}\right),
        \qquad
Q_\gamma(\omega)=\rho_\gamma(\omega)R(\omega).
$$
The branch of the logarithm was chosen by continuous phase unwrapping along the real interval.
The Newton iteration was stopped when the residual satisfied $|F_L(\omega)|<10^{-12}$.
This residual condition is used as a numerical stopping criterion, not as a proof of the absence of other zeros.
For $L=60.0$, the interval corresponds to $n=6,7,\ldots,38$, and all root searches converged from the above initial guesses.

As a numerical check of the regular--window assumption, we sampled the
real interval and the boundary of the rectangle
$$
U_{\rm num}
=
\{\,\omega=x+iy:\ 0.30\leq x\leq2.00,\ |y|\leq0.02\,\}.
$$
On the real interval we found $\min |A|=1.031$, $\min |Q_\gamma|=0.244$, and $\max |R|=0.916$.
On $\partial U_{\rm num}$ we found $\min |A|=1.028$, $\min |B|=0.246$, and $\max |R|=0.958$.
For the present value $\gamma=0$, the factor $i\omega-\gamma$ is also bounded away from zero on $U_{\rm num}$.
The numerical changes of argument of $A$ and $B$ along $\partial U_{\rm num}$ were zero.
This is only a diagnostic, not a proof of Assumption~\ref{assump:Q} or of the full nonvanishing condition \eqref{eq:full_prefactor_condition}.
It indicates, however, that the displayed region does not contain zeros or poles of the normalized round--trip coefficient and that the divided prefactor is regular in this numerical example.

Figure~\ref{fig:resonance_poles} compares the computed zeros with the
asymptotic centers. The agreement is already good for $L=60.0$, in
accordance with Theorem~\ref{thm:main_asymptotics}.

\begin{figure}[htbp]
\centering
\includegraphics[width=0.72\textwidth]{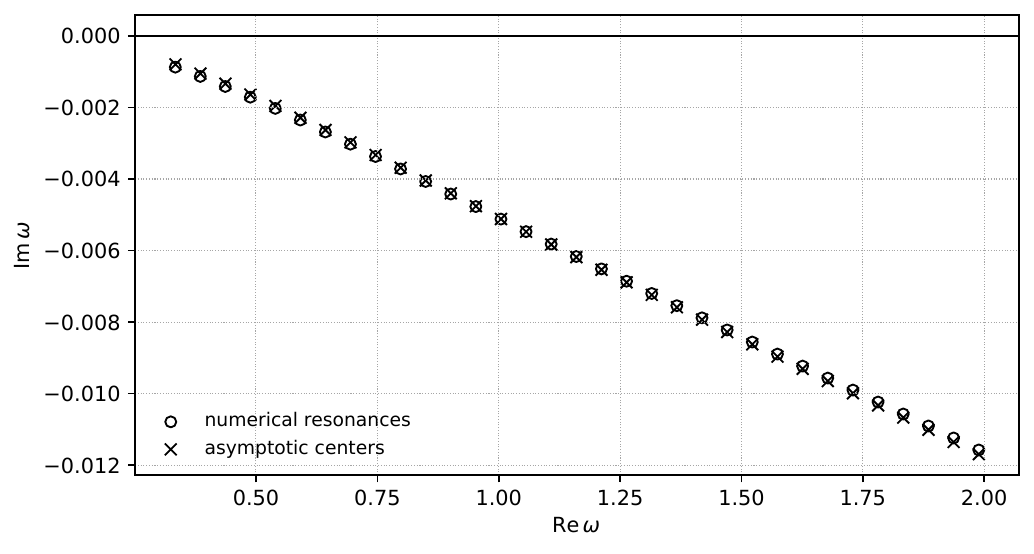}
\caption{
Numerically computed zeros and their asymptotic centers for the compactly
supported barrier with $L=60.0$, $\gamma=0.0$, $V_0=2.0$, and $a=1.0$.
The dots are the numerical zeros, and the crosses are the asymptotic
centers.
}
\label{fig:resonance_poles}
\end{figure}

The real parts of the resonances are almost equally spaced. This is
shown in Figure~\ref{fig:resonance_spacing}, where the numerical spacing
is compared with the leading value $\pi/L$. The small systematic
deviation is explained by the frequency dependence of the phase of $Q_\gamma$.
Indeed, if $q(x)=\log Q_\gamma(x)=a(x)+ib(x)$,
then, the real part of the asymptotic center is
$$
{\pi n\over L}-{b(\pi n/L)\over 2L}.
$$
Thus, the first finite--$L$ correction to the spacing is formally
$$
-{1\over 2L}
\left\{
b\left({\pi(n+1)\over L}\right)
-
b\left({\pi n\over L}\right)
\right\}
\simeq
-{\pi\over 2L^2}b'(x).
$$

\begin{figure}[htbp]
\centering
\includegraphics[width=0.72\textwidth]{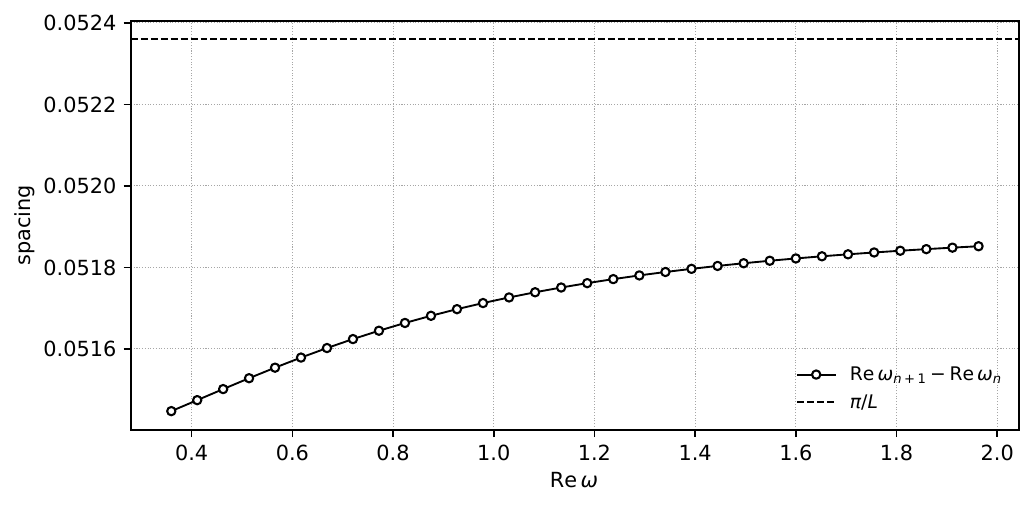}
\caption{
Spacing of the real parts of consecutive resonances.
The dashed line is $\pi/L$. The small deviation from this value is caused
by the frequency dependence of the phase of $Q_\gamma$.
}
\label{fig:resonance_spacing}
\end{figure}

We also checked the dependence on the cavity length. For
$L=40.0,60.0,80.0$, we used the same barrier parameters and the same real
frequency interval, and computed
$$
E_L=
\max_n
\left|
\omega_n-\omega_n^{(0)}
\right|,
$$
where the maximum is taken over the indices $n$ for which $\pi n/L$ lies
in the displayed interval. The results are shown in
Table~\ref{tab:L_dependence}. The approximate constancy of $L^2E_L$ is
consistent with the $O(L^{-2})$ error estimate in
Theorem~\ref{thm:main_asymptotics}.

\begin{table}[htbp]
\tbl{Dependence of the error
$E_L=\max_n|\omega_n-\omega_n^{(0)}|$ on the cavity length.
The column ``roots'' gives the number of roots in the displayed window.
\label{tab:L_dependence}}
{\begingroup
\renewcommand{\baselinestretch}{1}\selectfont
\renewcommand{\arraystretch}{0.9}
\setlength{\tabcolsep}{4pt}
\begin{tabular}{@{}cccc@{}}
\toprule
$L$ & roots & $E_L$ & $L^2E_L$ \\
\colrule
$40$ & $22$ & $8.23\times10^{-4}$ & $1.32$ \\
$60$ & $33$ & $3.70\times10^{-4}$ & $1.33$ \\
$80$ & $43$ & $2.09\times10^{-4}$ & $1.34$ \\
\botrule
\end{tabular}
\endgroup}
\end{table}

Next we visualize the same structure on a shifted real frequency line.
We use
$$
K_L^{\rm cav}(\omega)
=
{1\over
1-\rho_\gamma(\omega)R(\omega)e^{2i\omega L}},
        \qquad
\varepsilon=2.0\times10^{-4}.
$$
Figure~\ref{fig:transfer_comb} shows
$|K_L^{\rm cav}(\omega+i\varepsilon)|$. The peaks form an almost equally
spaced resonance comb, with spacing close to $\pi/L$.
\begin{figure}[htbp]
\centering
\includegraphics[height=0.27\textheight,keepaspectratio]{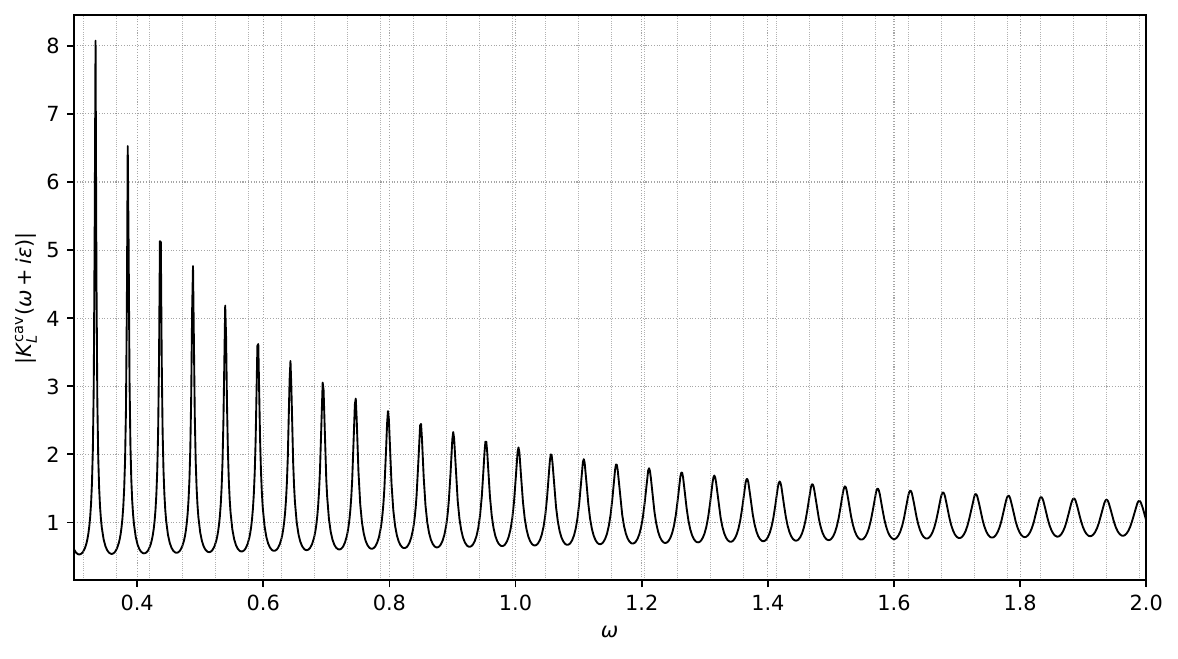}
\caption{
Shifted cavity transfer factor
$|K_L^{\rm cav}(\omega+i\varepsilon)|$ for
$L=60.0$, $\gamma=0.0$, and $\varepsilon=2.0\times 10^{-4}$.
The peaks form an almost equally spaced resonance comb, and the dotted
lines mark the leading centers $x_n=\pi n/L$.
}
\label{fig:transfer_comb}
\end{figure}
We finally illustrate the source factor. We first use two idealized
narrow pulses inside the free part of the cavity:
$$
S(y,\omega)
=
g(\omega)
\{\delta(y-y_1)+\delta(y-y_2)\}.
$$
Then \eqref{eq:D_def} gives
$$
D_L^{\delta}(\omega)
=
g(\omega)
\{\varphi_L(y_1,\omega)+\varphi_L(y_2,\omega)\}.
$$
For $-L<y_1,y_2<0$, and away from $\omega=0$,
$$
\varphi_L(y,\omega)
=
\cos\{\omega(y+L)\}
+
{\gamma\over \omega}\sin\{\omega(y+L)\}.
$$
Thus separated source components produce frequency--dependent
interference in the source factor. 
In the numerical plot we take
$$
g(\omega)
=
\exp\left\{-{(\omega-\omega_c)^2\over 2\sigma^2}\right\},
$$
with
$$
\omega_c=0.85,\qquad
\sigma=0.43,\qquad
y_1=-50.0,\qquad
y_2=-28.0.
$$
This is not an astrophysical initial perturbation, but a simple localized
source in the present half--line model.

Figure~\ref{fig:source_modulation} shows the resulting source modulation.
The comb remains, but its peak heights become strongly source dependent.
We also checked that replacing the delta functions by normalized smooth
bumps of width $h=0.4$ changes the normalized source factor by less than
$5.5\times10^{-3}$ on the displayed interval.

\begin{figure}[htbp]
\centering
\includegraphics[width=0.80\textwidth]{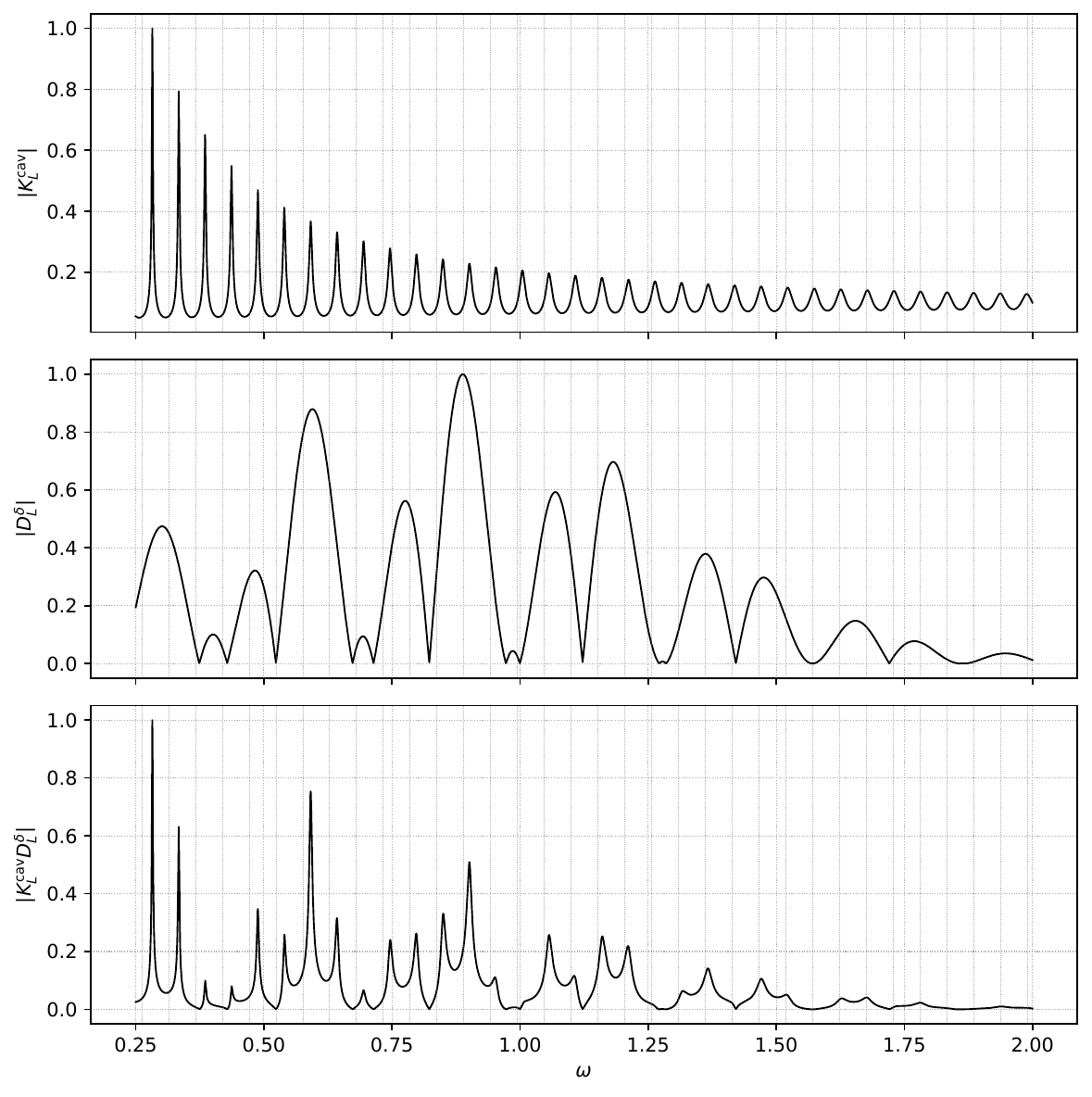}
\caption{
Source modulation of the resonance comb. The panels show the normalized
shifted cavity transfer factor $|K_L^{\rm cav}(\omega+i\varepsilon)|$,
the normalized shifted source factor
$|D_L^{\delta}(\omega+i\varepsilon)|$, and the normalized product
$|K_L^{\rm cav}(\omega+i\varepsilon)
D_L^{\delta}(\omega+i\varepsilon)|$. Each quantity is normalized by its
maximum in the displayed real frequency interval. The vertical dotted
lines mark the leading centers $x_n=\pi n/L$.
}
\label{fig:source_modulation}
\end{figure}

\rev{The source modulation in Figure~\ref{fig:source_modulation} is
related to, but should be distinguished from, the responses calculated
in Refs.~\refcite{Kokkotas1997,TominagaSaijoMaeda1999,FerrariKokkotas2000}.
Those works use concrete initial perturbations or particle sources in
relativistic stellar models.
Figure~\ref{fig:source_modulation}, in contrast, isolates only the
residue--level mechanism in the controlled half--line model.
Its purpose is to show that the homogeneous pole set does not by itself
determine the heights of the peaks in a source--to--observer response.}

We also give a direct numerical test of the pole--cancellation mechanism in Proposition~\ref{prop:visibility}.
This test is a functional--analytic illustration of the residue mechanism, not an astrophysical source model.
Let $\omega_j$ be the computed resonance with label $n=16$:
$$
\omega_j=0.849513-0.00406898\,i.
$$
Take two point sources at the same positions $y_1=-50.0$ and
$y_2=-28.0$, but now allow a complex relative coefficient,
$$
D_L^{\rm can}(\omega)
=
\varphi_L(y_1,\omega)
+
c_2\varphi_L(y_2,\omega),
        \qquad
c_2
=
-{\varphi_L(y_1,\omega_j)\over \varphi_L(y_2,\omega_j)}.
$$
For the present parameters this gives
$$
c_2=-1.22512-0.234108\,i.
$$
The common envelope $g(\omega)$ is omitted here, since any nonzero holomorphic factor does not affect the cancellation at $\omega_j$.
The result is shown in Table~\ref{tab:source_cancellation}.
The equal source has a nonzero source factor at $\omega_j$, while the canceling source has $D_L^{\rm can}(\omega_j)=0$ up to numerical roundoff.
This is the numerical version of the statement that the corresponding pole is removable in this source--to--observer response.

\begin{table}[htbp]
\tbl{Direct test of source--factor cancellation at the resonance
$\omega_j=0.849513-0.00406898\,i$.
%Here $D_L^{\rm eq}=\varphi_L(y_1,\omega)+\varphi_L(y_2,\omega)$ and
%$D_L^{\rm can}$ is defined above.
\label{tab:source_cancellation}}
{\begin{tabular}{@{}ccc@{}}
\toprule
source factor & value at $\omega_j$ & local shifted--line peak ratio \\
\colrule
$D_L^{\rm eq}$  & $|D_L^{\rm eq}(\omega_j)|=1.07539$ & $1$ \\
$D_L^{\rm can}$ & $|D_L^{\rm can}(\omega_j)|<2\times10^{-17}$ & $0.156$ \\
\botrule
\end{tabular}}
\end{table}

In the last column, the local peak ratio is the ratio of
$$
\max_{|\omega-\Rew\omega_j|\leq0.02}
\left|
K_L^{\rm cav}(\omega+i\varepsilon)D_L(\omega+i\varepsilon)
\right|
$$
to the same quantity for the equal source.
This number is not part of the theorem; it is affected by the shift,
the regular background, and nearby poles, and illustrates how a small
residue suppresses the peak on a shifted real frequency line.
%%%% 

\subsection{\textcolor{black}{Untruncated Regge--Wheeler comparison}}
\label{subsec:rw_numerics}

\begingroup
\color{black}
We next test the resonance approximation for a physically standard
black--hole perturbation potential.
We consider the axial Regge--Wheeler potential
\refcite{ReggeWheeler1957,Chandrasekhar1983},
$$
V_{\rm RW}(r)
=
\left(1-{2M\over r}\right)
\left\{
{\ell(\ell+1)\over r^2}
-
{6M\over r^3}
\right\},
$$
where $M$ is the Schwarzschild mass and $\ell$ is the angular
multipole number.  The tortoise coordinate is
$$
r_*
=
r+2M\log\left({r\over2M}-1\right).
$$
We write $x$ for $r_*$ after choosing the additive constant so that the
maximum of $V_{\rm RW}$ is at $x=0$.
The wall is placed at $x=-L$, and we impose the same Robin condition
as in the analytic model.
In the calculations below we set
$$
M=1,\qquad \ell=2,\qquad \gamma=0,
$$
and use the frequency window
$$
0.18\leq M\Rew\omega\leq0.34.
$$
This window is bounded away from the threshold and lies below the
maximum barrier frequency
$$
M\sqrt{V_{\rm RW}^{\max}}\simeq0.389.
$$

Let $f_+(x,\omega)$ denote the solution which is outgoing at spatial
infinity.
The direct Regge--Wheeler resonances are calculated as the zeros of
$$
F_{{\rm RW},L}(\omega)
=
f_+'(-L,\omega)-\gamma f_+(-L,\omega).
$$
At a large but finite outer point $x_{\rm right}$, the outgoing
condition is imposed using the first WKB logarithmic derivative.
More precisely, we use
$$
{f_+'(x_{\rm right},\omega)\over
 f_+(x_{\rm right},\omega)}
=
ik(x_{\rm right},\omega)
+
{V_{\rm RW}'(x_{\rm right})\over
4k(x_{\rm right},\omega)^2},
\qquad
k(x,\omega)^2=\omega^2-V_{\rm RW}(x),
$$
where the branch of $k$ is chosen continuously from $k=\omega$.
The irrelevant overall amplitude is fixed by
$$
f_+(x_{\rm right},\omega)=1.
$$
The solution is then integrated backwards to the wall by an adaptive
Runge--Kutta method.

On the horizon side, the same outgoing solution has the asymptotic form
$$
f_+(x,\omega)
=
A_{\rm RW}(\omega)e^{i\omega x}
+
B_{\rm RW}(\omega)e^{-i\omega x},
\qquad x\to-\infty.
$$
For the numerical extraction of these coefficients, we use
$x_{\rm left}/M=-60$ and the plane--wave formulas
$$
A_{\rm RW}(\omega)
=
{1\over2}e^{-i\omega x_{\rm left}}
\left\{
f_+(x_{\rm left},\omega)
+{f_+'(x_{\rm left},\omega)\over i\omega}
\right\},
$$
$$
B_{\rm RW}(\omega)
=
{1\over2}e^{i\omega x_{\rm left}}
\left\{
f_+(x_{\rm left},\omega)
-{f_+'(x_{\rm left},\omega)\over i\omega}
\right\}.
$$
Moving the extraction point to $x_{\rm left}/M=-80$ leaves the
displayed values unchanged.
We define
$$
R_{\rm RW}(\omega)
=
{B_{\rm RW}(\omega)\over A_{\rm RW}(\omega)}
$$
and compare the direct resonances with the first asymptotic centers
$$
\widehat{\omega}_n
=
{\pi n\over L}
+
{i\over2L}
\log\left\{
\rho_\gamma\left({\pi n\over L}\right)
R_{\rm RW}\left({\pi n\over L}\right)
\right\}.
$$
The phase of $R_{\rm RW}$ is unwrapped continuously along the real
frequency window to fix the branch of the logarithm.
We denote the direct zeros of $F_{{\rm RW},L}$ by
$\omega_n^{\rm RW}$ and obtain them by a Newton--type iteration in
$\Rew\omega$ and $\Imw\omega$, starting from
$\widehat{\omega}_n$.
The rigorous theorem of Section~\ref{sec:asymptotics} is not being
applied to the untruncated potential.
The centers $\widehat{\omega}_n$ are used here as a numerical
large--cavity approximation whose accuracy is tested directly.

Figure~\ref{fig:rw_resonance_comparison} compares the direct
Regge--Wheeler resonances with the asymptotic centers for
$L/M=40,60,80$.
The same nearly equally spaced comb is visible for all three cavity
lengths.
The deviations are larger in the imaginary parts than in the real
parts, but both decrease as the cavity length increases.
\endgroup

\begin{figure}[htbp]
\color{black}
\centering
\includegraphics[width=0.98\textwidth]
{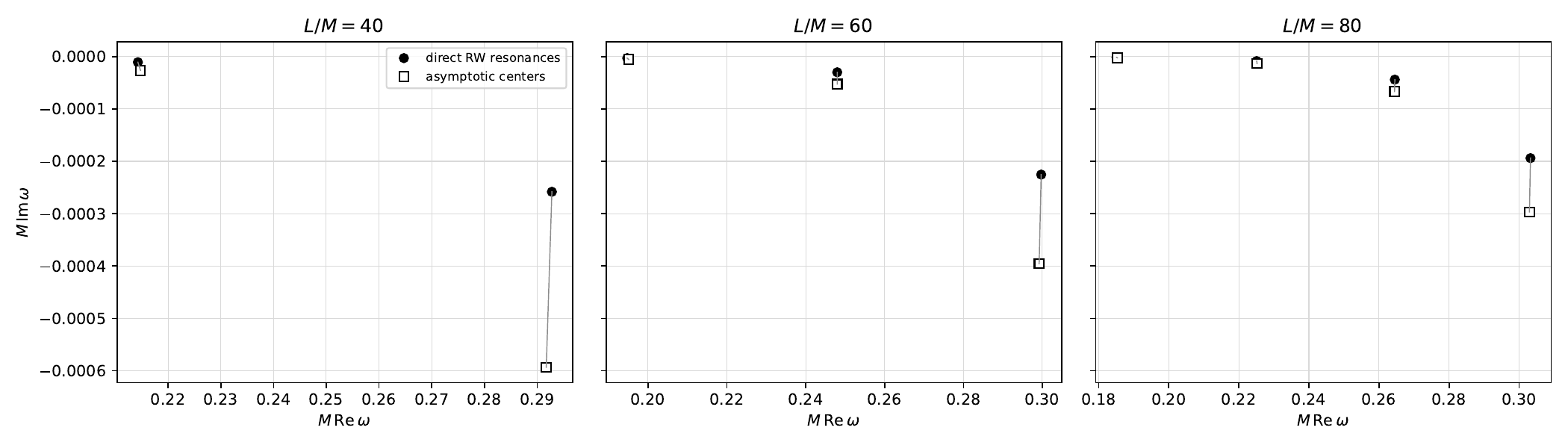}
\caption{\color{black}
Direct resonances for the untruncated axial Regge--Wheeler potential
and their first asymptotic centers.
Filled circles denote the zeros of $F_{{\rm RW},L}$, and open squares
denote $\widehat{\omega}_n$.
The gray segments connect corresponding modes.
The parameters are $M=1$, $\ell=2$, and $\gamma=0$.
}
\label{fig:rw_resonance_comparison}
\end{figure}

\begingroup
\color{black}
For each cavity length, we define
$$
E_L^{\rm RW}
=
\max_n
\left|
\omega_n^{\rm RW}-\widehat{\omega}_n
\right|,
$$
where the maximum is taken over the modes in the displayed frequency
window.
The results are given in Table~\ref{tab:rw_error}.
The quantity $L^2E_L^{\rm RW}$ remains between $1.79$ and $2.06$.
This approximate constancy is numerically consistent with an
$L^{-2}$ localization error also for the untruncated Regge--Wheeler
potential.
It is not asserted here as a theorem for the long--range potential.
\endgroup

\begin{table}[htbp]
\color{black}
\tbl{\color{black}
Comparison between the direct Regge--Wheeler resonances and the first
asymptotic centers.
The column ``roots'' gives the number of modes in the displayed
frequency window.
\label{tab:rw_error}}
{\begingroup
\color{black}
\renewcommand{\baselinestretch}{1}\selectfont
\renewcommand{\arraystretch}{0.9}
\setlength{\tabcolsep}{5pt}
\begin{tabular}{@{}cccc@{}}
\toprule
$L/M$ & roots & $E_L^{\rm RW}$ & $L^2E_L^{\rm RW}$ \\
\colrule
$40$ & $2$ & $1.1164\times10^{-3}$ & $1.7863$ \\
$60$ & $3$ & $5.4530\times10^{-4}$ & $1.9631$ \\
$80$ & $4$ & $3.2167\times10^{-4}$ & $2.0587$ \\
\botrule
\end{tabular}
\endgroup}
\end{table}

\begingroup
\color{black}
We checked the stability of the outgoing numerical condition by
repeating the $L/M=60$ calculation with
$$
x_{\rm right}/M=200,\quad250,\quad300.
$$
For the three displayed modes, the changes between
$x_{\rm right}/M=250$ and $300$ were
$$
1.66\times10^{-12},\qquad
5.71\times10^{-12},\qquad
4.13\times10^{-11},
$$
respectively.
The largest of these changes is more than seven orders of magnitude
smaller than $E_{60}^{\rm RW}$.
With the normalization $f_+(x_{\rm right},\omega)=1$, the absolute
residuals
$$
\left|F_{{\rm RW},L}(\omega_n^{\rm RW})\right|
$$
were at most of order $10^{-12}$.
Thus the differences in Table~\ref{tab:rw_error} are not explained by
the finite outer boundary used in the numerical integration.

For example, the real parts of the three $L/M=60$ resonances are
$$
0.194732653,\qquad
0.247991038,\qquad
0.299712432.
$$
Their consecutive spacings are
$$
0.053258385,\qquad0.051721393,
$$
while
$$
{\pi\over L}=0.052359878.
$$
The direct Regge--Wheeler resonances therefore exhibit the expected
nearly equally spaced comb, including the small phase--dependent
deviations from $\pi/L$.
The magnitudes of the imaginary parts increase with frequency,
consistently with increasing leakage through the exterior barrier.

We finally test source dependence directly in the Regge--Wheeler
problem.
Let $\varphi_L(x,\omega)$ be the solution satisfying the wall
normalization
$$
\varphi_L(-L,\omega)=1,
\qquad
\varphi_L'(-L,\omega)=\gamma.
$$
For a smooth source $S_j$, we calculate
$$
D_{L,j}^{\rm RW}(\omega)
=
\int_{-L}^{0}
\varphi_L(y,\omega)S_j(y)\,dy.
$$
Since the wall--normalized solution $\varphi_L$ and the outgoing
solution $f_+$ have Wronskian $F_{{\rm RW},L}$, we define the
Regge--Wheeler transfer factor at the numerical outer point by
$$
K_L^{\rm RW}(\omega)
=
-{1\over F_{{\rm RW},L}(\omega)}.
$$
We use the normalized Gaussian
$$
G(y;y_0,h)
=
{1\over\sqrt{2\pi}h}
\exp\left\{
-{(y-y_0)^2\over2h^2}
\right\}
$$
with $h=1.2$ and compare
$$
S_1(y)=G(y;-0.75L,h)
$$
with
$$
S_2(y)
=
{1\over2}
\left\{
G(y;-0.75L,h)
+
G(y;-0.45L,h)
\right\}.
$$
Thus the two profiles have the same integrated source strength.
The two source factors are normalized by one common scale, and the two
responses are also normalized by one common scale.
Their relative magnitudes are therefore retained.

Figure~\ref{fig:rw_source_modulation} shows the result for $L/M=60$.
For this plot, the functions are evaluated at
$\omega+i\varepsilon$, where
$\varepsilon=2.0\times10^{-4}$.
The poles of the homogeneous transfer factor are unchanged, but their
spectral weights depend strongly on the source.
Near the first resonance, the two separated source components interfere
destructively and almost remove the corresponding response peak.
Near the higher resonances, the relative excitation is different.
This gives a direct Regge--Wheeler realization of the residue--level
source dependence derived in Section~\ref{sec:response}.
\endgroup

\begin{figure}[htbp]
\color{black}
\centering
\includegraphics[width=0.76\textwidth]
{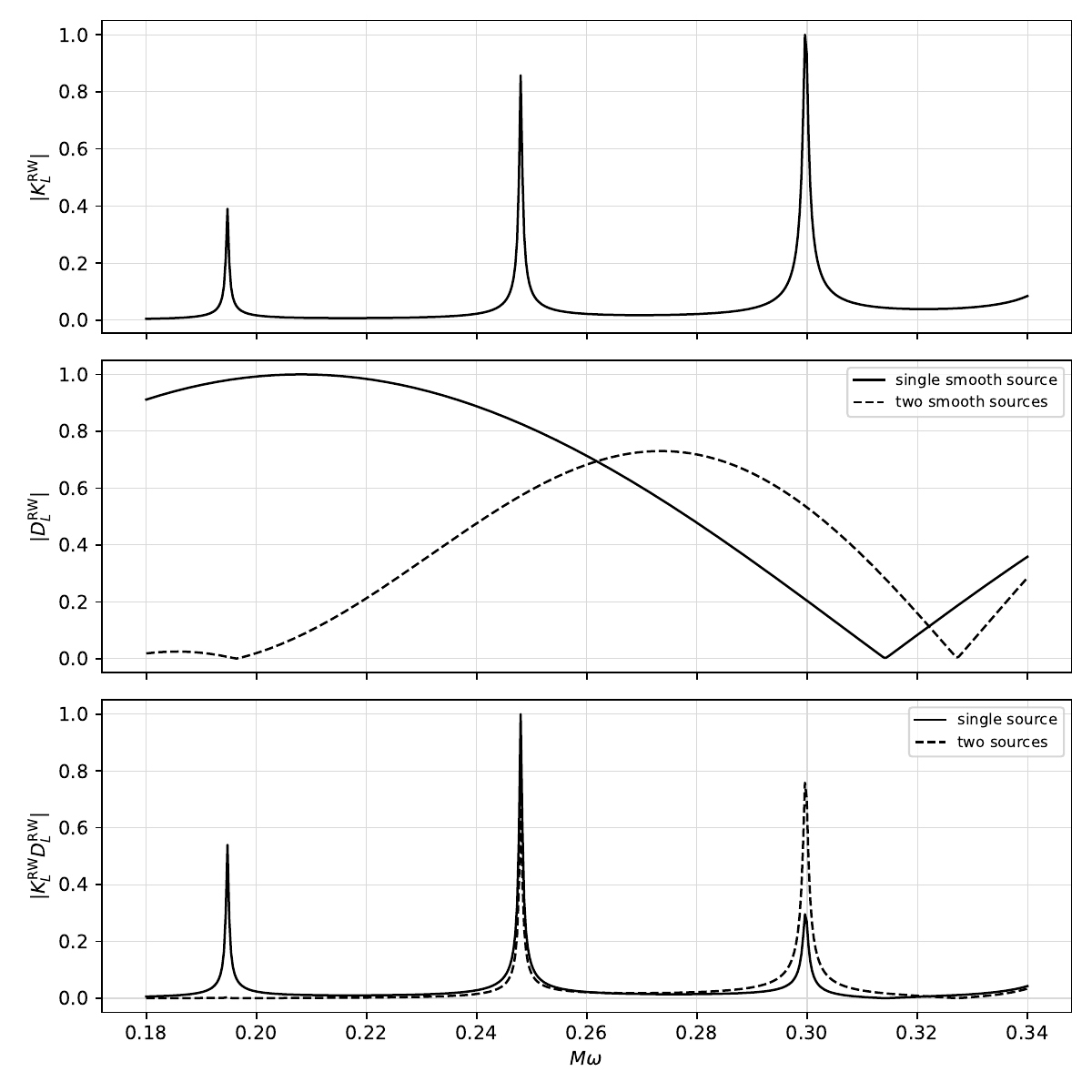}
\caption{\color{black}
Source modulation for the untruncated axial Regge--Wheeler potential
with $L/M=60$, $\ell=2$, $\gamma=0$, and
$\varepsilon=2.0\times10^{-4}$.
The upper, middle, and lower panels show
$|K_L^{\rm RW}(\omega+i\varepsilon)|$,
$|D_{L,j}^{\rm RW}(\omega+i\varepsilon)|$, and
$|K_L^{\rm RW}(\omega+i\varepsilon)
D_{L,j}^{\rm RW}(\omega+i\varepsilon)|$, respectively.
The transfer factor is normalized by its maximum, while the two source
factors are normalized by one common scale and the two products are
normalized by one common scale.
Solid and dashed lines denote single and double smooth source profiles
with the same integrated strength.
The pole locations are source independent, whereas their spectral weights
depend strongly on the source profile.
}
\label{fig:rw_source_modulation}
\end{figure}

\FloatBarrier
\section{Discussion}
\label{sec:discussion}

We have studied black--hole echo resonance spectra in a controlled half--line transfer--function model with an exterior compact barrier and an inner reflecting wall.
The Fabry--Perot type denominator is not proposed as a new physical mechanism; the contribution is to derive it from the Jost solution and the Robin wall condition, construct its zeros in a regular window with an $O(L^{-2})$ error bound, and identify them with the local resonances of the full half--line problem when the divided prefactor is nonzero.

The source factorization separates the homogeneous resonance spectrum from the spectral response to a particular excitation.
For the inhomogeneous problem, the response has the form $\psi_\omega=K_L(\omega)D_L(\omega)$, where $K_L$ is the Wronskian transfer factor and $D_L$ is a source factor.
Thus a homogeneous pole may give a prominent peak, a suppressed peak, or a canceled pole in a particular source--to--observer response.
In this sense, source dependence is a residue--level effect.
It should not be confused with detectability in gravitational--wave data, which also involves waveform modeling, detector noise, and statistical analysis.
The pole--cancellation example in Section~\ref{sec:numerics} is included only to illustrate this analytic mechanism, not as an astrophysical source model.

The local spacing obtained here is the frequency--domain counterpart of the echo delay.
A long cavity of length approximately $L$ gives the leading angular--frequency spacing $\Delta\Rew\omega\sim\pi/L$, or $\Delta f\sim1/(2L)$ in ordinary frequency.
This is equivalent to the usual time--domain relation $\Delta t_{\rm echo}\sim2L$.
Lower order corrections come from the frequency dependence of the round--trip phase.
Thus the spacing is determined by the cavity length and reflection phases, whereas the observed peak heights also depend on the source factor and the regular background.

The regular--window assumption is the main technical condition in the local theorem.
It excludes the threshold $\omega=0$ and zeros or poles of the normalized round--trip coefficient $Q_\gamma(\omega)=\rho_\gamma(\omega)R(\omega)$, so that a holomorphic branch of $\log Q_\gamma$ can be chosen.
When the normalized zeros are identified with full half--line resonances, the divided prefactor $(i\omega-\gamma)A(\omega)e^{-i\omega L}$ must also be holomorphic and nonzero.
Near exceptional points, such as zeros of $R(\omega)$ or poles of $R(\omega)=B(\omega)/A(\omega)$, the normalized equation is no longer the appropriate local object, and one should return to the undivided boundary equation.
The numerical checks in Section~\ref{sec:numerics} are therefore diagnostics for the displayed example, while the theorem relies on the stated holomorphic nonvanishing assumptions.

The imaginary parts of the local resonances are controlled by the logarithm of the total round--trip reflection coefficient.
In the Robin model, the wall reflection coefficient has unit modulus on the real axis, so damping comes from leakage through the exterior barrier.
The same local argument applies to a more general effective wall as soon as $Q_{\rm wall}(\omega)=\rho_{\rm wall}(\omega)R(\omega)$ is holomorphic and nonzero in the chosen complex neighborhood, as stated in Corollary~\ref{cor:general_wall}.
Physical admissibility of a prescribed wall reflectivity, for example passivity or causality, is a separate question and is not needed for the local resonance--comb theorem proved here.

\rev{The compact support of $V$ is the main simplification in the
rigorous part of the paper.
It gives exact free plane waves on both sides of the barrier and a
meromorphic reflection coefficient away from $\omega=0$.
Accordingly, the one--zero--per--cell theorem and its $O(L^{-2})$
localization estimate are proved only under the compact--support
assumption.

The numerical results of Section~\ref{subsec:rw_numerics} nevertheless
show that the principal conclusions are not artifacts of cutting off
the black--hole potential.
For the untruncated axial Regge--Wheeler potential, the direct
wall--plus--barrier resonances form the same nearly equally spaced comb,
the first asymptotic centers give quantitatively accurate
approximations, and the observed errors for $L/M=40,60,80$ are
consistent with $L^{-2}$ scaling.
The Wronskian response also continues to separate the homogeneous pole
locations from the source--dependent spectral weights.

These calculations do not provide a long--range analogue of
Theorem~\ref{thm:main_asymptotics}.
In particular, the analytic continuation, threshold behavior, and
uniform error estimates for the exact Regge--Wheeler or Zerilli
potentials require a separate analysis.
The numerical comparison instead shows which structural conclusions
survive in a physically standard perturbation model.}

\rev{The compactly supported examples and the Regge--Wheeler calculation
play different roles.
The former verify the theorem under its stated assumptions and
illustrate exact pole cancellation.
The latter test the physical robustness of the same mechanism for an
untruncated black--hole perturbation potential.
In both cases, the homogeneous denominator fixes the pole locations,
while the source factor controls their residues and the resulting peak
weights.
This distinction is consistent with earlier calculations of
source--dependent mode excitation in relativistic stellar models
\refcite{Kokkotas1997,TominagaSaijoMaeda1999,FerrariKokkotas2000}
and with the recent plunge--spectrum calculation of Nair
\refcite{Nair2026}.}

\rev{The physical origin of the cavity need not be unique.
For example, the double--peak potential found for black holes in
massive gravity can generate echoes even in the presence of an event
horizon
\refcite{DongStojkovic2021}.
By contrast, the plunge calculation of Ref.~\refcite{Nair2026}
provides a concrete excitation process for the low--frequency
resonance comb of a black--hole mimicker.
Together, these examples emphasize two separate points of the present
analysis: the homogeneous cavity structure determines the resonance
locations, while the physical source determines how strongly those
resonances appear in a particular spectrum.}

Rotating and superradiant situations are outside the present theorem.
In such problems the effective reflection coefficients may have a more complicated frequency dependence, and superradiant amplification can change the stability question.
If the modulus of an effective round--trip coefficient exceeds one in a frequency range, the location of the resonances and the possible appearance of instabilities have to be analyzed separately.
The present result should therefore be viewed as a controlled nonrotating benchmark for the resonance comb and the source--dependence mechanism.

\rev{The pseudospectral stability of the resulting non--self--adjoint
resonance problem is a separate question; see
Ref.~\refcite{JaramilloMacedoAlSheikh2021}.}

\rev{In summary, this paper gives a controlled analytic treatment of
black--hole echo resonance spectra.
It proves a local resonance comb with an $O(L^{-2})$ localization error
for a compactly supported barrier and shows how the source factor
controls the residues of the corresponding response.
Direct calculations for the untruncated axial Regge--Wheeler potential
show that the same comb, the first asymptotic centers, and the
source--dependent peak modulation remain quantitatively relevant in a
standard gravitational perturbation model.
The rigorous treatment of long--range black--hole potentials, as well
as extensions to rotating and superradiant systems, remains open.}

\end{document}